\documentclass[letterpaper, 10 pt]{ieeetran}  % Comment this line out if you need a4paper
\IEEEoverridecommandlockouts                              % This command is only needed if you want to use the \thanks command%\documentclass[journal,twoside,web]{ieeecolor}
\usepackage{graphics} % for pdf, bitmapped graphics files
\usepackage{epsfig} % for postscript graphics files
\usepackage{amsmath} % assumes amsmath package installed
\usepackage{amssymb}  % assumes amsmath package installed
\usepackage{enumitem}
\usepackage{booktabs}
\usepackage{graphicx}

\usepackage{cite}

\usepackage{amsthm}
\newtheorem{assumption}{\bf Assumption}

\newtheorem{theorem}{\bf Theorem}
\newtheorem{proposition}{\bf Proposition}
\newtheorem{lemma}{\bf Lemma}
\newtheorem{corollary}{\bf Corollary}
\newtheorem{remark}{\bf Remark}

\usepackage{setspace}

\usepackage{tikz} \usetikzlibrary{arrows,patterns,mindmap,backgrounds,shadows}
\usetikzlibrary{backgrounds}
\usetikzlibrary{shapes,arrows,chains}

\usepackage{siunitx}
\usepackage[colorlinks=true,        % links are colored
  citecolor=black,        % color of cite links
  linkcolor=blue,        % color of hyperref links
  menucolor=green,        % color of Acrobat Reader menu buttons
  urlcolor=blue,         % color of urls \url{}
  bookmarks=true, hyperindex=true, breaklinks=true]{hyperref}
\usepackage{comment,cancel}
\usepackage{algpseudocode, algorithm}
\usepackage{enumitem}

\DeclareMathOperator*{\argmin}{arg\,min}

\def\bquad{\!\!\!\!}

\begin{document} 
\title{Adaptive Economic Model Predictive Control: Performance Guarantees for Nonlinear Systems}
\author{Maximilian Degner$^{1,2}$, Raffaele Soloperto$^1$, Melanie N. Zeilinger$^2$, John Lygeros$^1$, Johannes K\"ohler$^2$
\thanks{$^1$Automatic Control Laboratory, ETH Zürich}
\thanks{$^2$Institute for Dynamic Systems and Control, ETH Zürich\\
mdegner(at)ieee(dot)org, \{soloperr$|$mzeilinger$|$lygeros$|$jkoehle\}(at)ethz(dot)ch\\
Johannes K\"ohler and Raffaele Soloperto were supported as a part of NCCR Automation, a National Centre of Competence in Research, funded by the Swiss National Science Foundation (grant number 51NF40\_225155).\\
Raffaele Soloperto was also supported by
the European Research Council under the H2020 Advanced Grant no. 787845 (OCAL).}
}

\maketitle
%%%%%%%%%%%%%%%%%%%%%%%%%%%%%%%%%%%%%%%%%%%%%%%%%%%%%%%%%%%%%%%%%%%%%%%%%%%%%%%%
\begin{abstract}
    We consider the problem of optimizing the economic performance of nonlinear constrained systems subject to uncertain time-varying parameters and bounded disturbances. 
In particular, we propose an adaptive economic model predictive control (MPC) framework that: (i) directly minimizes transient economic costs, (ii) addresses parametric uncertainty through online model adaptation, (iii) determines optimal setpoints online, and (iv) ensures robustness by using a tube-based approach.  
The proposed design ensures recursive feasibility, robust constraint satisfaction, and a transient performance bound. In case the disturbances have a finite energy and the parameter variations have a finite path length, the asymptotic average performance is (approximately) not worse than the performance obtained when operating at the best reachable steady-state.  
We highlight performance benefits in a numerical example involving a chemical reactor with unknown time-invariant and time-varying parameters.

\end{abstract}  

\section{Introduction}\label{sec:intro}
Model Predictive Control (MPC) is an optimization-based control technique that can ensure constraint satisfaction for general nonlinear systems~\cite{rawlingsModelPredictiveControl2020}.
The performance of MPC schemes relies on the accuracy of the model, which is typically identified offline~\cite{darby2012mpc}. 
However, the system dynamics often change during online operation. Thus, online model adaptation is required to ensure consistent performance~\cite{yoonAdaptivePredictiveControl1994, astromAdaptiveControl2008}.

In addition, the performance of a plant is usually characterized by some economic cost, such as the yield of chemical plants~\cite{rawlingsFundamentalsEconomicModel2012,ellisEconomicModelPredictive2017,faulwasserEconomicNonlinearModel2018} or the energy cost in building temperature control~\cite{maApplicationEconomicMPC2014}, and power grids~\cite{hu2023economic}.
However, most MPC implementations consider this economic cost indirectly  by stabilizing an optimal setpoint. 
In this work, we address both limitations by minimizing the transient economic cost directly and adapting the prediction model online.

\begin{figure}[t]
    \centering
    \includegraphics[width=0.9\linewidth]{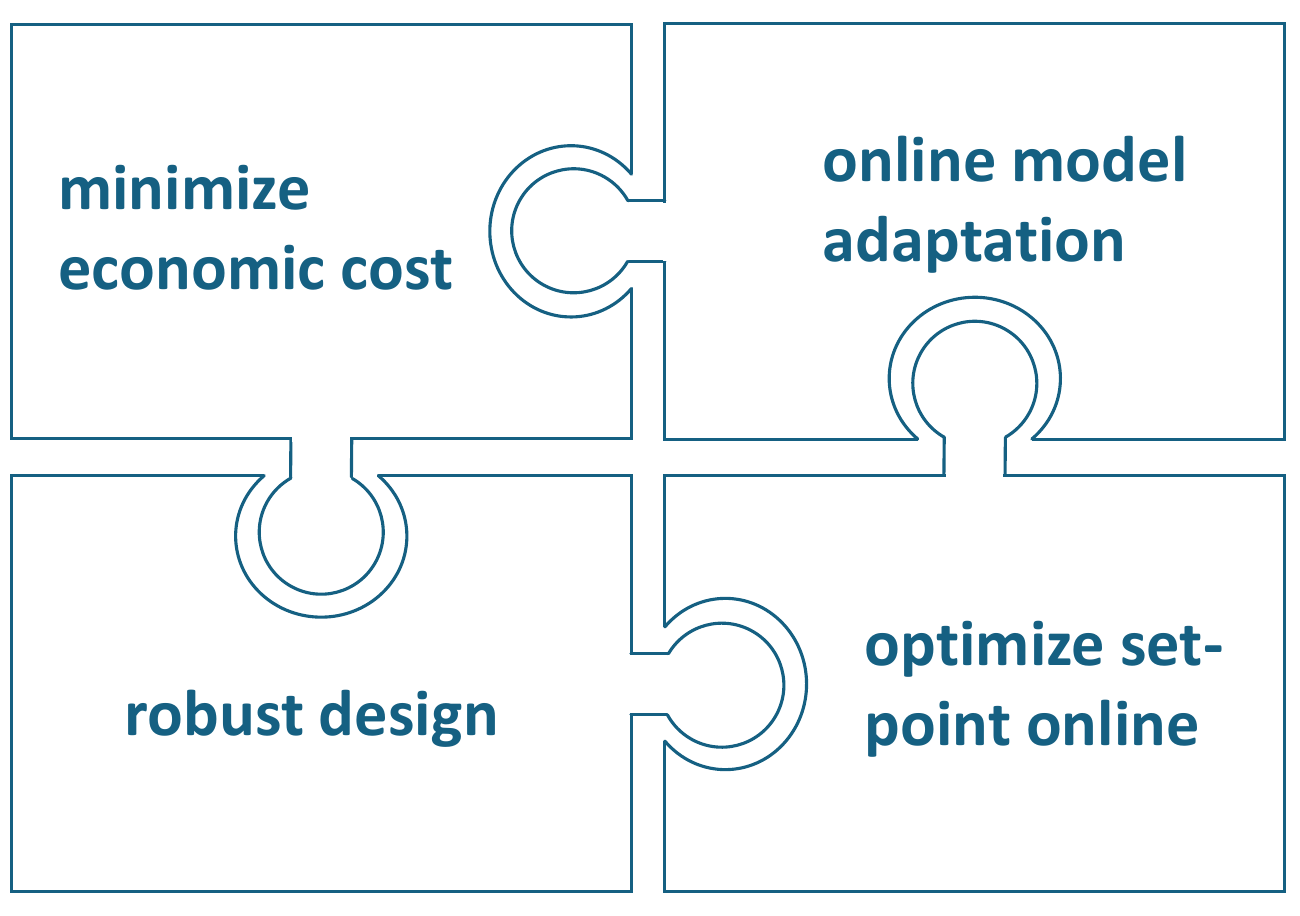}
    \caption{%
    We minimize the economic cost and adapt the prediction model during online operation. 
    Utilizing existing robust MPC designs guarantees constraint satisfaction and 
    thanks to artificial references,
    optimize setpoints online, leveraging the adapted model.
    }
    \label{fig:intro}
\end{figure} 

\subsection*{Related Work}
\emph{Adaptive MPC} schemes deal with online changes in the system dynamics by adapting the model parameters during closed-loop operation. 
While empirical results demonstrating performance enhancement through adaptive MPC are well documented~\cite{yoonAdaptivePredictiveControl1994, bouffard2012learning,kumarExperimentalEvaluationMIMO2015, didierRobustAdaptiveModel2021}, theoretical developments have been lacking until recently~\cite{mayne2014model}. 
\emph{Robust adaptive MPC} approaches for linear and nonlinear uncertain systems are proposed in~\cite{lorenzenRobustMPCRecursive2019, kohlerLinearRobustAdaptive2019} and \cite{adetola2011robust,lopez2019adaptive,kohlerRobustAdaptiveModel2021,sinha2022adaptive,sasfiRobustAdaptiveMPC2023,solopertoGuaranteedClosedLoopLearning2023}, respectively. These approaches ensure robust constraint satisfaction while leveraging online adaptation of the model uncertainty to reduce the conservatism.  
In~\cite{aswani2013provably} the optimized cost leveraged a separate prediction with an online learned model to enhance performance. 
By adapting such a model with a least-mean square parameter estimator, adaptive MPC schemes have shown $\mathcal{L}_2$ stability of the closed-loop system despite parametric uncertainty~\cite{lorenzenRobustMPCRecursive2019, kohlerLinearRobustAdaptive2019,kohlerRobustAdaptiveModel2021}. 
However, all these approaches rely on a cost that is positive definite compared to fixed equilibrium. 
Hence, these approaches cannot directly optimize economic costs or consider online changes in the optimal setpoint.
In~\cite{guayAdaptiveEconomicOptimising2013} and~\cite[Chap.~3]{wuLearningBasedEconomicControl2023}, also economic cost functions have been investigated in an adaptive MPC scheme. However, worst-case costs are optimized and no guarantees regarding closed-loop performance are established.  

\emph{Economic MPC} formulations~\cite{rawlingsFundamentalsEconomicModel2012,faulwasserEconomicNonlinearModel2018} directly minimize the predicted economic cost, such as the energy consumption of heating, ventilation, and cooling (HVAC) systems~\cite{taheriModelPredictiveControl2022}, the production yield of a chemical reactor~\cite{rawlingsFundamentalsEconomicModel2012}, or the profit of flexible manufacturing systems~\cite{risbeckUnificationClosedloopScheduling2019}.  This direct optimization of the economic cost can significantly enhance the performance compared to traditional stabilizing MPC formulations~\cite{ellisEconomicModelPredictive2017}.
% thereby enhancing the overall performance~\cite{rawlingsFundamentalsEconomicModel2012,faulwasserEconomicNonlinearModel2018, ellisEconomicModelPredictive2017}.
In the nominal case, i.e., perfect model knowledge and no disturbances, economic MPC theory ensures that the asymptotic average performance of the closed-loop system is not worse than the performance of operating at the optimal steady-state~\cite{angeliAveragePerformanceStability2012,amritEconomicOptimizationUsing2011}.  
In case the optimal steady-state is subject to online changes, artificial references can be included in the economic MPC formulation to optimize setpoints online~\cite{krupaModelPredictiveControl2024,fagianoGeneralizedTerminalState2013, mullerEconomicModelPredictive2013,mullerPerformanceEconomicModel2014, ferramoscaEconomicMPCChanging2014, kohlerPeriodicOptimalControl2020}. In particular, such approaches can ensure that the asymptotic average performance is (approximately) not worse than the performance at the best reachable steady-state~\cite{fagianoGeneralizedTerminalState2013, mullerEconomicModelPredictive2013, mullerPerformanceEconomicModel2014}.
Although these economic MPC schemes yield strong performance guarantees, the results rely on a perfect prediction model and are thus not applicable to uncertain systems. 
In case of model mismatch, \emph{robust} economic MPC schemes address this issue by optimizing for the worst-case cost~\cite{bayerOptimalSystemOperation2018}. % additional references if needed: bayerTubebasedRobustEconomic2014,schwenkelRobustEconomicModel2020,kloppeltTransientPerformanceTubebased2021
While this provides robust bounds on the closed-loop performance, it can also lead conservative operation.

To maximize the closed-loop economic performance of uncertain systems, we recently proposed an \emph{adaptive economic} MPC scheme that uses least-mean squares parameter adaptation and provides strong asymptotic and transient performance guarantees~\cite{degnerAdaptiveEconomicModel2024}.
However, the approach is limited to open-loop stable linear systems, soft state constraints, and assumes that the optimal steady-state is independent of the unknown model parameters.

\subsection*{Contribution}
In this paper, we propose a robust adaptive economic MPC scheme that provides suitable bounds on the closed-loop cost. 
We consider nonlinear systems that are subject to state and input constraints, time-varying uncertain parameters, and bounded disturbances. 
We have the following main contributions:
\begin{itemize}
    \item We derive a transient performance guarantee relative to the optimized steady-state of the certainty-equivalent prediction model (Theorem~\ref{thm_transient}).
    \item In case of finite-energy disturbances and finite parameter variations, we show that the asymptotic average closed-loop cost is (approximately) not worse than the cost at the optimal achievable steady-state (Corollaries~\ref{corr_asymptotic_performance} and~\ref{corr_subopt}).
    \item We show how to design an economic terminal cost that can deal with online changing model parameters and setpoints and does not rely on a sufficiently small terminal set (Proposition~\ref{prop_constr-term-ingred}).
\end{itemize}

Figure~\ref{fig:intro} depicts the main components of the adaptive economic MPC formulation: The model parameters are adapted online using least-mean square parameter adaptation. The MPC directly minimizes the predicted economic cost using a certainty-equivalent prediction with the online parameter estimates. As the optimal steady-state is unknown a priori, the terminal cost and set are formulated using online optimized artificial references. We utilize a tube-based MPC scheme to guarantee constraint satisfaction for all admissible disturbances and parameters.
Based on simulations of a chemical reactor, we demonstrate the performance benefits of the proposed adaptive economic MPC scheme.

\paragraph*{Outline}
In Section~\ref{sec:problem_setup}, we describe the control problem and discuss preliminaries. 
In Section~\ref{sec:proposed_design}, we present the proposed adaptive economic MPC approach, while the theoretical analysis and discussion can be found in Section~\ref{sec:theo_analysis}. We illustrate the performance benefits of the proposed approach with a numerical example in Section~\ref{sec:num_ex} and conclude in Section~\ref{sec:concl}. Appendices~\ref{app:proofs} and~\ref{appendix_aux-results} contain the proofs of our main results and of auxiliary results, respectively. In Appendix~\ref{ssec:app_SDP}, we provide details details regarding the offline design.

\subsection*{Notation}
For any $a,b\in\mathbb{R}$, we denote $\mathbb{I}_{[a,b]} = {\{n\in\mathbb{N}| a\leq n\leq b\}}$, where $\mathbb{N}$ signifies the non-negative integers. A positive semidefinite matrix $R$ is denoted by $R\succeq 0$ and a positive definite matrix $Q$ is denoted by $Q\succ0$. We denote the $1$-norm and $2$-norm of a vector $x\in\mathbb{R}^n$ by $\|x\|_1$ and $\|x\|$, respectively. 
For a matrix $Q\succ0$, we denote the weighted norm by $\|x\|_Q^2=x^\top Qx$. 
A sequence $w_k\in\mathbb{R}^n$, $k\in\mathbb{I}_{[0,T]}$ is denoted by $w_{[0,T]}$, and its $\mathcal{L}_2$-norm by  $|w_{[0,T]}|_{\mathcal{L}_2}^2:={\sum_{k=0}^T\|w_k\|^2}$. 
The spectral norm of a matrix $A\in\mathbb{R}^{n\times m}$ is given by $\|A\|=\sqrt{\lambda_{\text{max}}(A^\top A)}$, where $\lambda_{\max}$ denoted the maximal eigenvalue.
The identity matrix of size $n\times n$ is signified by $I_n$.
The $i$-th element of a vector $x\in\mathbb{R}^n$ is represented by $[x]_i$.
A function $\alpha$ is of class $\mathcal{K}_\infty$, if $\alpha: \mathbb{R}_{\geq0}\rightarrow\mathbb{R}_{\geq0}, \alpha(0)=0, \lim_{s\rightarrow \infty}\alpha(s)=\infty$, and $\alpha$ is strictly increasing and continuous.

\section{Problem Setup and Preliminaries} 
\label{sec:problem_setup}
In this section, we first discuss the problem setup (Section~\ref{ssec:setup_system}). Then, we present background on robust MPC (Section~\ref{ssec:robust_MPC}) and least-mean square parameter adaptation (Section~\ref{ssec:LMS}). Lastly, we introduce regularity conditions on the system (Section~\ref{ssec:syst_prop}).

\subsection{Problem Setup} \label{ssec:setup_system}
\label{sec:system}
The system dynamics are given by
\begin{align}\label{prob_setup:eq-dynamics}
\begin{split}
    x_{k+1} &= f_0(x_k,u_k)+G(x_k,u_k)\theta_k + E(x_k,u_k)d_k \\
    &=: f_\mathrm{w}(x_k,u_k,\theta_k,d_k),
    \end{split}
\end{align} 
where the states, control inputs, disturbances, the parameters, and the time are denoted by $x_k\in\mathbb{R}^{n_x},\ u_k\in\mathbb{R}^{n_u}$, $d_k \in\mathbb{R}^{n_d}$, $\theta_k\in\mathbb{R}^{n_\theta}$, and $k\in\mathbb{N}$, respectively. The dynamics $f_\mathrm{w}$ are assumed to be continuous, while the time-varying parameters $\theta_k$ is assumed to be unknown.
The considered model structure, which is linear in the parameters $\theta_k$, is standard in adaptive control~\cite{astromAdaptiveControl2008, krsticNonlinearAdaptiveControl1995}. 
In case uncertain model parameters enter non-linearly in the dynamics, we can often obtain the structure~\eqref{prob_setup:eq-dynamics} by suitably re-defining the parameters $\theta$. 
We denote the nominal dynamics, i.e., neglecting disturbances $d$, by
\begin{align} \label{nonlin_discrete_system-eq}
    f(x,u,\theta) := f_0(x,u)+G(x,u)\theta.
\end{align}

\begin{assumption}[Bounded disturbances and parameters]
\label{assump_bounded-disturb-param}
    The disturbances and parameters are contained in known sets $\mathbb{D},\ \Theta$, i.e., $d_k\in\mathbb{D}$ and $\theta_k\in\Theta$ for all $k\in\mathbb{N}$. Furthermore, $\mathbb{D}$ is compact and includes the origin, and $\Theta$ is compact and convex.
\end{assumption}

The goal is to minimize the economic cost $\ell:{\mathbb{R}^{n_x}\times\mathbb{R}^{n_u}} \rightarrow \mathbb{R}$ of the closed-loop system, which directly reflects the economic performance and is generally not simply penalizing the distance to some optimal steady-state(s). 
In addition, state and input constraints should be robustly satisfied
\begin{equation} \label{prob_setup:constraints-eq}
    (x_k,u_k)\in \mathbb{Z},\quad \forall k\in\mathbb{N},
\end{equation} 
where $\mathbb{Z}$ is a compact set. 
The proposed approach addresses this problem by incorporating online parameter adaptation and a robust design in the economic MPC formulation. 
Thus, in the following, we introduce preliminaries regarding robust MPC and parameter adaptation.

\subsection{Robust MPC framework} \label{ssec:robust_MPC}
In the following, we recapitulate tube-based robust MPC, which ensures that the constraints~\eqref{prob_setup:constraints-eq} are satisfied for all possible realizations of the uncertain model~\eqref{prob_setup:eq-dynamics}. Analogous to~\cite{kohlerRobustAdaptiveModel2021}, we consider general conditions on the tube, which we will utilize in the MPC design to robustly ensure constraint satisfaction and which can be constructively satisfied by most existing tube-based robust MPC approaches (see Remark~\ref{rk:tube}). 

The considered robust MPC formulation predicts the evolution of a nominal trajectory $(z_{j|k},v_{j|k})$ using the nominal model~\eqref{nonlin_discrete_system-eq} with some nominal parameters $\bar{\theta}\in\Theta$. Here, $z_{j|k}$ denotes the prediction of $z$ for time $k+j$ computed at time $k$.
In addition, the MPC predicts a tube $\mathbb{X}$ that incorporates all possible trajectories resulting from the parametric uncertainty and disturbances affecting the system, i.e., $x\in\mathbb{X}\ \forall d\in \mathbb{D},\theta\in\Theta$.
Furthermore, we make use of an auxiliary feedback $u=\kappa(x,z,v)$ 
that drives the state $x$ to the nominal prediction $z$, and therefore reduces the size of the predicted tube. Without loss of generality, we consider $\kappa(z,z,v)=v$, $\forall (z,v)\in\mathbb{Z}$.  
Formally, the tube $\mathbb{X}$ is propagated via a general mapping $\mathbb{X}_{k+1}\supseteq\Phi(\mathbb{X}_k,z_k,v_k)$ that satisfies the following assumption.

\begin{assumption}[Tube propagation]
\label{assump_Phi-function}
     For any $\mathbb{X}\subseteq\mathbb{R}^{n_{x}}$, any $x,z\in\mathbb{X}$, any $(x,\kappa(x,z,v))\in\mathbb{Z}, (z,v)\in\mathbb{Z}$, any $\theta \in \Theta$, and any $d\in\mathbb{D}$, the tube propagation mapping $\Phi$ satisfies
    \begin{align} \label{assump_Phi-function-eq}
        f_\mathrm{w}(x,\kappa(x,z,v),\theta,d) \in \ \Phi(\mathbb{X},z,v).
    \end{align}
\end{assumption}
Property~\eqref{assump_Phi-function-eq} enables efficient implementations by sequentially propagating a set containing all possible state trajectories, see also Figure~\ref{fig:tube_illustration} for an illustration.
\begin{remark}
\label{rk:tube}[Tube designs]
The general condition~\eqref{assump_Phi-function-eq} can be satisfied by many sequential propagation techniques, such as interval arithmetic~\cite{limonRobustMPCConstrained2005, meyerIntervalReachabilityAnalysis2021}, rigid tubes~\cite{yuTubeMPCScheme2013,bayerDiscretetimeIncrementalISS2013,singh2023robust}, or Lipschitz based propagation techniques~\cite{kohlerNovelConstraintTightening2018,pinRobustModelPredictive2009}, see~\cite{houskaRobustOptimizationMPC2019,althoffSetPropagationTechniques2021} for an overview.

One of the simplest and computationally most efficient designs include a linear feedback $\kappa(x,z,v) = K(x-z)+v$
and scaled tube of the form $\mathbb{X}_{j|k}= \{z_{j|k}\} \oplus s_{j|k} \mathbb{E}$, where $\mathbb{E}$ is an offline specified set (typically an ellipsoid or polytope/zonotope), and $s_{j|k}$ is a scaling that is automatically predicted in the MPC~\cite{adetola2011robust,lopez2019adaptive,kohlerRobustAdaptiveModel2021,sasfiRobustAdaptiveMPC2023}. 
In Appendix~\ref{ssec:app_SDP}, we also provide explicit formulas to design such a set $\mathbb{E}$ and determine the scaling by adapting the results from~\cite{kohlerComputationallyEfficientRobust2021,sasfiRobustAdaptiveMPC2023}.  
\end{remark}

Given Assumption~\ref{assump_Phi-function}, we can formulate a general robust MPC problem as follows:\begin{subequations}\label{robust-MPC-scheme}\begin{alignat}{2} 
    \bquad \min_{\substack{z_{\cdot|k},\, v_{\cdot|k},\\ v_k^\mathrm{s},\, \mathbb{X}_{\cdot|k}}}&   && \mathrm{cost}
    \label{robust-MPC_scheme-a}\\[-5pt]
     &\ \ \text{s.t.}  ~ && z_{j+1|k} = f (z_{j|k}, v_{j|k}, \bar{\theta}),\label{robust-MPC_scheme-b}\\
                    &    &&\mathbb{X}_{j+1|k} \supseteq \Phi(\mathbb{X}_{j|k}, z_{j|k}, v_{j|k}),\label{robust-MPC_scheme-c}\\
                    &    && \{(x,\kappa(x, z_{j|k}, v_{j|k})):~ x\in \mathbb{X}_{j|k}\} \subseteq\mathbb{Z},  \label{robust-MPC_scheme-d}\\
                    &    && z_{0|k}, x_k \in \mathbb{X}_{0|k},\label{robust-MPC_scheme-e}\\
                    &    && \mathbb{X}_{N|k} \supseteq \Phi(\mathbb{X}_{N|k}, z_{N|k}, v_k^\mathrm{s}),\label{robust-MPC_scheme-f}\\
                    & && \{(x,\kappa(x, z_{N|k}, v_k^\mathrm{s})):~ x\in \mathbb{X}_{N|k}\} \subseteq\mathbb{Z}, \label{robust-MPC_scheme-f2}\\
                    &    && z_{N|k} = f(z_{N|k}, v_k^\mathrm{s}, \bar{\theta}),\label{robust-MPC_scheme-g}\\
                    &    && \qquad \qquad \forall j \in\mathbb{I}_{[0,N-1]}.\nonumber
\end{alignat}\end{subequations}
We minimize some suitable cost function~\eqref{robust-MPC_scheme-a}, while propagating the nominal system dynamics~\eqref{robust-MPC_scheme-b}, propagating the tube~\eqref{robust-MPC_scheme-c}, and ensuring robust constraint satisfaction~\eqref{robust-MPC_scheme-d}. The set $\mathbb{X}_{j|k}$ can be finitely parameterized (e.g., ellipsoid or polytope). The constraint~\eqref{robust-MPC_scheme-e} ensures that the initial tube $\mathbb{X}_{0|k}$ contains the measured state $x_k$ and the optimized nominal initial state $z_{0|k}$. The nominal trajectory is constrained to end at a steady-state with some input $v^\mathrm{s}_k$~\eqref{robust-MPC_scheme-g}. The constraint~\eqref{robust-MPC_scheme-f} ensures that the terminal set $\mathbb{X}_{N|k}$ is robust positive invariant and~\eqref{robust-MPC_scheme-f2} ensures that it also satisfies the constraints. 

We denote the solution of Problem~\eqref{robust-MPC-scheme} by $z_{j|k}^\ast,\ v_{j|k}^\ast$, and $\mathbb{X}_{\cdot|k}^\ast$.

\begin{lemma}[Robust MPC properties~\protect{\cite[Theorem~1]{kohlerRobustAdaptiveModel2021}}]
    Let Assumptions~\ref{assump_bounded-disturb-param} and~\ref{assump_Phi-function} hold. If Problem~\eqref{robust-MPC-scheme} is feasible at time $k=0$, then it is recursively feasible and the closed-loop system $x_{k+1} = f_\mathrm{w}(x_k, u_k, \theta_k, d_k)$  with $u_k=\kappa(x_k,z^\ast_{0|k},v^\ast_{0|k})$ satisfies the constraints~\eqref{prob_setup:constraints-eq}.
\end{lemma}

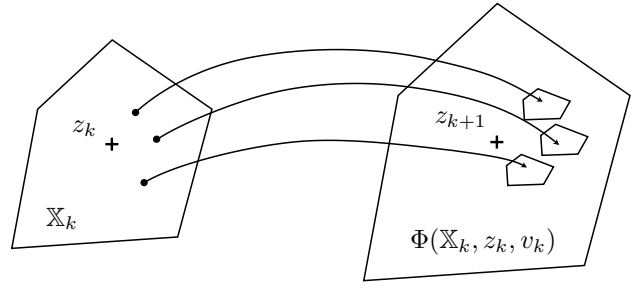
\begin{figure}[t]
    \centering
    \begin{tikzpicture}[scale = 1]
        \node at (0,0) {\includegraphics[width=0.95\linewidth]{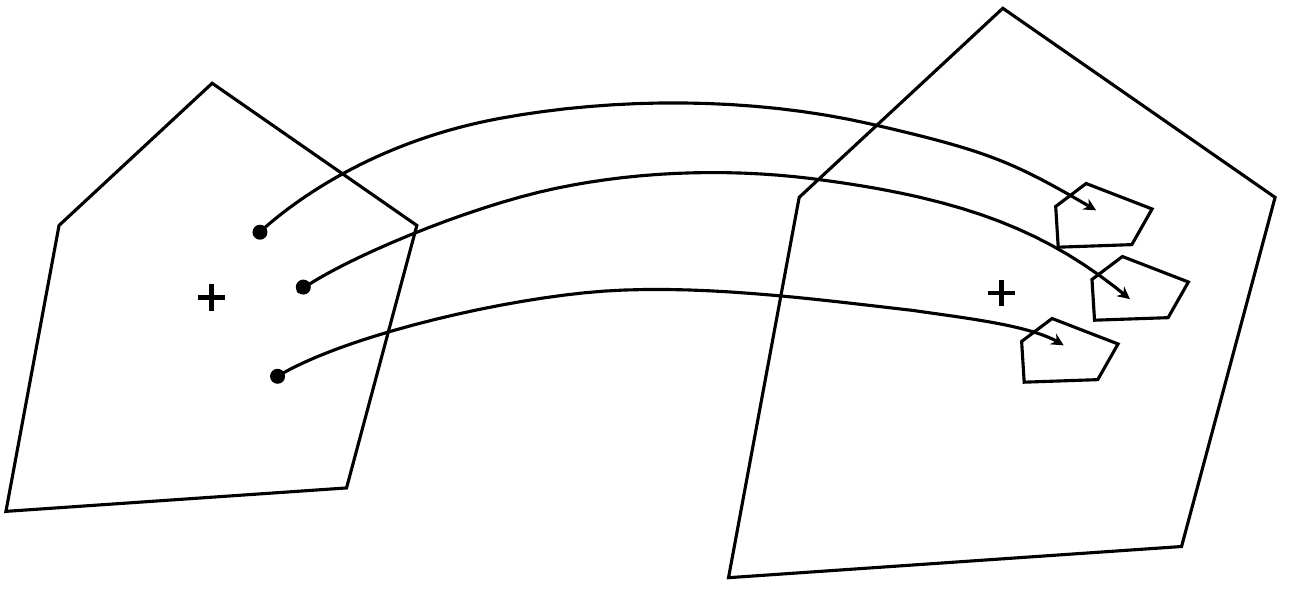}};
        \node at (-3.5,-1.0) {$\mathbb{X}_{k}$};
        \node at (2.1,-1.3) {$\Phi(\mathbb{X}_{k}, z_{k}, v_{k})$};
        \node at (1.8,0.3) {$z_{k+1}$};
        \node at (-3.2,0.2) {$z_{k}$};
    \end{tikzpicture}
    \caption{For each $x\in \mathbb{X}_{k}$, we can predict future states for different disturbances $d\in \mathbb{D}$ and parameters $\theta\in\Theta$. All of these predictions are contained in the robust one-step reachable set $\Phi(\mathbb{X}_{k}, z_{k}, v_{k})$.}
    \label{fig:tube_illustration}
\end{figure}

%%%%%%%%%%%%%%
\subsection{Least-Mean Squares Parameter Estimation} \label{ssec:LMS}
In the following, we recall the projected least-mean squares (LMS) estimator from~\cite{lorenzenRobustMPCRecursive2019}. This parameter estimation method iteratively computes an estimate $\hat{\theta}_k\in\Theta$ starting from the initial estimate $\hat{\theta}_0\in\Theta$. The estimate $\hat{\theta}_k$ will be used in our proposed adaptive economic MPC formulation.

To shorten the notation, we define $w_k := E(x_k,u_k) d_k$ and $\Delta\theta_k := \theta_{k+1} - \theta_k$.
We define the one-step prediction $\hat{x}_{1|k}$ and the corresponding prediction error $\tilde{x}_{1|k}$ as 
\begin{subequations}
\begin{align} 
    \hat{x}_{1|k} &:= f_0(x_k,u_k) + G(x_k,u_k) \cdot \hat{\theta}_k,             \label{LMS_one-step_prediction-eq}\\
        \tilde{x}_{1|k} &:= G(x_k,u_k) \cdot (\theta_k-\hat{\theta}_k).     \label{probsetup_prop_LMS-prediction-error-eq}
\end{align}
\end{subequations}
Note that the one-step prediction error satisfies $\tilde{x}_{1|k} + w_k= x_{k+1}-\hat{x}_{1|k}$.

The update equations of the projected LMS estimator are
\begin{subequations}\label{LMS_update-eqs}\begin{align} 
\begin{split}
        \tilde{\theta}_{k+1} 
    &= \hat{\theta}_k + \mu G(x_k,u_k)^\top \cdot\left(x_{k+1}-\hat{x}_{1|k}\right),
    \end{split}\label{LMS_update_a}\\
    \hat{\theta}_{k+1} &= \argmin_{\theta \in \Theta} \|\theta-\tilde{\theta}_{k+1}\|, \label{LMS_update_b}
\end{align} \end{subequations}
where $\mu>0$ is the parameter update gain and satisfies 
\begin{align}\label{probsetup_LMS_gain-definition}
    \frac{1}{\mu} \geq \max_{(x,u)\in \mathbb{Z}} \|G(x,u)\|^2.
\end{align}
Equation~\eqref{LMS_update_a} corresponds to a gradient step on the instantaneous squared prediction error, while~\eqref{LMS_update_b} projects the intermediate result onto the convex set $\Theta$ to guarantee $\hat{\theta}_{k+1} \in \Theta$.
Since $G$ is continuous and $\mathbb{Z}$ compact  (cf. Section~\ref{ssec:setup_system}), the maximum in~\eqref{probsetup_LMS_gain-definition} is bounded. 
\begin{proposition}[LMS bounds]\label{probsetup_prop_LMS-bound}
    Let Assumption~\ref{assump_bounded-disturb-param} hold. If $(x_k,u_k)\in \mathbb{Z}$, $\forall k\in\mathbb{N}$, then, for all $T\in\mathbb{N}$, it holds that
    \begin{align}\label{probsetup_prop_LMS-bound-eq}
        \sum_{k=0}^T \|\tilde{x}_{1|k}\|^2 &\leq \frac{1}{\mu}\|\hat{\theta}_0-\theta_0\|^2 + \sum_{k=0}^T \left[\|w_k\|^2 
        +\frac{2c_\theta}{\mu} \|\Delta \theta_k\|\right],
    \end{align}
    where $c_\theta:=\max_{\theta_1,\theta_2\in\Theta}\|\theta_1-\theta_2\|$.
    Furthermore, the difference between successive parameter estimates satisfies
    \begin{equation} \label{probsetup_prop_LMS-theta-diff-eq}
    \frac{1}{\sqrt{\mu}} \|\hat{\theta}_{k+1}-\hat{\theta}_k\| \leq \|\tilde{x}_{1|k}+w_k\|.
    \end{equation}
\end{proposition}

\begin{proof}
    The proof follows the steps in~\cite[Lemma~5]{lorenzenRobustMPCRecursive2019} and is provided in Appendix~\ref{app:proofs}.
\end{proof}
Note that Proposition~\ref{probsetup_prop_LMS-bound} holds without assuming that a persistence of excitation condition holds. Furthermore, it only provides a bound on the prediction error, not on the error in the parameter estimate $\|\hat{\theta}_k-\theta_k\|$.

\subsection{Regularity Conditions} \label{ssec:syst_prop}
We impose mild regularity conditions on the system dynamics, the cost, and the auxiliary feedback law. To leverage convex design principles even when dealing with non-convex constraints, we define the convex and compact sets $\mathbb{Z}_x$ and $\mathbb{Z}_u$ such that
\begin{align}\label{def_Zu_Zx-eq}
    (x,v)\in\mathbb{Z}\Rightarrow x\in\mathbb{Z}_x,v\in\mathbb{Z}_u.
\end{align}
In case $\mathbb{Z}$ is convex, these two sets can simply be the projection of $\mathbb{Z}$ on the state and input space.
\begin{assumption}[Regularity] \label{assump_regularity}
The dynamics $f_{\mathrm{w}}$, the cost $\ell$, and the stabilizing feedback $\kappa$ 
are Lipschitz continuous, i.e., there exist constants $L_{\mathrm{dyn}},L_{\ell}, L_{\kappa} \geq 0$, such that the inequalities
\begin{subequations}\label{assump_regularity-eq}
\begin{align}
    \begin{split} 
    &\bquad \| f_\mathrm{w}(x,u,\theta,d)-f_\mathrm{w}(x',u',\theta',d')\|\\
    &\bquad \leq L_\mathrm{dyn} (\|x-x'\| + \|u-u'\| + \|\theta-\theta'\| + \|d-d'\|),
    \label{assump_regularity-eq-a}
    \end{split}\\[5pt]
    %%%%%%%%%%%%%%%
    & \bquad \|\ell(x,u)-\ell(x',u')\| \leq L_{\ell}(\|x-x'\|+\|u-u'\|),
    \label{assump_regularity-eq-b}\\[5pt]
    %%%%%%%%%%%%%%%
    \begin{split}
    &\bquad \|\kappa(x,z,v)-\kappa(x',z',v')\| \\
    &\bquad\bquad \qquad\qquad\quad\  \leq L_\kappa \left(\|x-x'\| + \|z-z'\| + \|v-v'\| \right), \label{assump_regularity-eq-c}
    \end{split}
\end{align}
\end{subequations}
hold for all $(x,u), (x',u'), (z,v), (z',v')\in\mathbb{Z}_x\times \mathbb{Z}_u$, and all $(\theta,d), (\theta',d')\in \Theta\times \mathbb{D}$.
\end{assumption}
This assumption is satisfied for continuously differentiable functions $f_\mathrm{w},\ell,\kappa$, given that $\mathbb{Z},\Theta, \mathbb{D}$ are compact (cf. Assumption~\ref{assump_bounded-disturb-param} and~\eqref{prob_setup:constraints-eq}). 

Furthermore, we impose a mild regularity condition on the equilibrium manifold of system~\eqref{nonlin_discrete_system-eq}.
\begin{assumption}[Steady-state manifold]\label{assump_equi-manifold}
    There exists a function $h:\, \mathbb{Z}_u\times \Theta \rightarrow \mathbb{R}^{n_x}$ 
    such that for any $\theta \in \Theta,\ v\in \mathbb{Z}_u$, the states $x=h(v,\theta),\ z^\mathrm{s}=h(v,\bar{\theta})$ satisfy
    \begin{align*} 
        \begin{split}
        &f (x,\kappa(x,z^\mathrm{s},v), \theta) = x, \qquad f(z^\mathrm{s},v,\bar{\theta})=z^\mathrm{s}.
        \end{split}
    \end{align*}
    Furthermore, $h$ is Lipschitz continuous, i.e., there exists a constant $L_\mathrm{h}\geq 0$, such that 
    \begin{align}
    \label{assump_equilib-manifold-eq}
        \|h(v,\theta)-h(v',\theta')\| &\leq L_\mathrm{h} (\|v-v'\| + \|\theta-\theta'\|),
    \end{align}
    holds for any $v,v'\in\mathbb{Z}_u,\ \theta,\theta'\in\Theta$.
\end{assumption}
Note that the true parameters $\theta_k$ are not known and Assumption~\ref{assump_equi-manifold} requires that, for each parameter, at least one equilibrium exists.
Such mapping $h$ exists naturally due to the implicit function theorem if the feedback $\kappa$ ensures exponential stability, see~\cite[Remark~1]{limonNonlinearMPCTracking2018}. 
\section{Proposed MPC Design} \label{sec:proposed_design}
In this section, we present our proposed adaptive economic MPC scheme. First, we discuss the elements of the MPC formulation (Section~\ref{ssec:ae-mpc_approach}) and then show how to design the terminal cost (Section~\ref{ssec:offline_design}).

\subsection{MPC Design}\label{ssec:ae-mpc_approach}
Our proposed robust adaptive economic MPC approach consists of three elements: 
\begin{itemize}
    \item We use the parameter estimate $\hat{\theta}_k$ from the projected LMS~\eqref{LMS_update-eqs} to obtain a certainty-equivalent prediction of the system's behavior $\hat{x}_{\cdot|k}$. The MPC minimizes the finite-horizon economic cost of this trajectory, consisting of the economic costs $\ell$ and the terminal cost $\ell_\mathrm{f}$.
    \item  Additionally, the economic cost at an artificial equilibrium  $(\hat{x}^\mathrm{s}_k,\hat{u}^\mathrm{s}_k)$ is minimized with some weighting $\beta\geq0$.
    \item To capture the effects of the parametric uncertainty and disturbances, we employ the robust MPC formulation from Section~\ref{ssec:robust_MPC}.
\end{itemize}

The robust adaptive economic MPC problem is given by
\begin{subequations}\label{MPC_scheme}
\begin{alignat}{2} 
    \bquad\!\! \min_{\substack{\hat{x}_{\cdot|k},\,z_{\cdot|k},\\v_{\cdot|k},\,v_k^{s},\\\mathbb{X}_{\cdot|k}}}  &
             \ \sum_{j=0}^{N-1}&& \ell(\hat{x}_{j|k},\,\hat{u}_{j|k})+\ell_\mathrm{f}(\hat{x}_{N|k},\hat{x}^\mathrm{s}_k,v^\mathrm{s}_k, \hat{\theta}_k) \nonumber\\[-15pt]
            & +\beta\, &&\ell(\hat{x}^\mathrm{s}_k,\hat{u}^\mathrm{s}_k)
            \label{MPC_scheme-a}\\[5pt]
     &\text{s.t.}&&\text{\eqref{robust-MPC_scheme-b}--\eqref{robust-MPC_scheme-g}}, \nonumber \\
    & && \hat{x}_{j+1|k} = f(\hat{x}_{j|k},\hat{u}_{j|k},\hat{\theta}_k), \label{MPC_scheme-d}\\
    & && \hat{u}_{j|k}=\kappa(\hat{x}_{j|k}, z_{j|k}, v_{j|k}),  \label{MPC_scheme-e}\\
    & && \hat{x}_{0|k} = x_k, \label{MPC_scheme-g}\\ 
    & && \hat{x}^\mathrm{s}_k = f(\hat{x}^\mathrm{s}_k, \hat{u}^\mathrm{s}_k,\hat{\theta}_k), \ \hat{u}^\mathrm{s}_k = \kappa(\hat{x}^\mathrm{s}_k,z_{N|k},v^\mathrm{s}_k),
    \label{MPC_scheme-h}\\
    & && \ell(\hat{x}^\mathrm{s}_k,\hat{u}^\mathrm{s}_k)\leq \lambda_k , \label{MPC_scheme-j}\\
    & && \qquad \qquad \forall j\in\mathbb{I}_{[0,N-1]}.\nonumber
\end{alignat}
\end{subequations}

The solution of~\eqref{MPC_scheme} is denoted by $\hat{x}_{\cdot|k}^\ast,\mathbb{X}_{\cdot|k}^\ast, z_{\cdot|k}^\ast, v_{\cdot|k}^\ast, v_k^{s\ast}, \hat{x}^\mathrm{s\ast}_k, \hat{u}^\mathrm{s\ast}_k$, and the corresponding value of the cost function~\eqref{MPC_scheme-a} by $\mathcal{J}^\ast_N(x_k,\hat{\theta}_k, \lambda_k)$. The closed-loop system is obtained by applying the control input 
\begin{align} \label{MPC-scheme_optimal-input}
    u_k = \kappa(x_k,z_{0|k}^\ast, v_{0|k}^\ast)
\end{align} to the system~\eqref{prob_setup:eq-dynamics}.

Similar to the robust MPC formulation~\eqref{robust-MPC-scheme}, we predict a nominal trajectory~\eqref{robust-MPC_scheme-d} with the nominal parameters $\bar{\theta}$ and a prediction tube~\eqref{robust-MPC_scheme-b} that satisfies the constraints for all parameters and disturbances~\eqref{robust-MPC_scheme-c}.
In contrast to the robust MPC scheme~\eqref{robust-MPC-scheme}, a certainty-equivalent economic cost is minimized.  
The corresponding trajectory $\hat{x}_{j|k}$ is the prediction of the state evolution for parameters $\hat{\theta}_k$ without disturbances. Furthermore, we  account for the estimated parameters by also minimizing the cost at the artificial reference $(\hat{x}^s_k,\hat{u}^\mathrm{s}_k)$ and for the long-term costs with a terminal cost $\ell_{\mathrm{f}}$.
The terminal cost $\ell_{\mathrm{f}}$ bounds the infinite-horizon cost of applying the feedback $\hat{u}_{j|k}=\kappa(\hat{x}_{j|k},z_{N|k},v^\mathrm{s}_k)$, $j\in\mathbb{I}_{\geq N}$ after the prediction horizon.  
The constant $\beta>0$ in the cost function~\eqref{MPC_scheme-a} weighs the artificial reference cost relative to the stage and terminal cost.

The function $h$ from Assumption~\ref{assump_equi-manifold} provides a mapping from the parameter estimate $\hat{\theta}_{k}$ and the \emph{previous} steady-input $v^\mathrm{s\ast}_{k-1}$ to a candidate steady-state $\hat{x}^\mathrm{s}_{k}=h(v^\mathrm{s\ast}_{k-1},\hat{\theta}_{k})$ with the corresponding input $\hat{u}^\mathrm{s}_{k} = \kappa(\hat{x}^\mathrm{s}_{k}, z^\mathrm{\ast}_{N|k-1}, v^\mathrm{s\ast}_{k-1})$.
We use this to update the value of $\lambda_k$ at each time step $k\in\mathbb{N}$ according to
\begin{align}
\label{eq:lambda_k}
    \lambda_k:=\ell(h(v^\mathrm{s\ast}_{k-1},\hat{\theta}_k),\hat{u}^\mathrm{s}).
\end{align} 
In combination with~\eqref{MPC_scheme-j}, this ensures that the cost at the artificial steady-state is not increased compared to the candidate solution.

%%%%%%%%%%%%%%%%%%%%%%%%%%%%%%%%%%%%%%%%%%%%%%%%%%%%%%%%%%%%%%%%%%%%%%%%%%%%%%%%%%%%%%%%%%%%%%%%%%%%%%%%%%%%%%

\subsection{Offline Design} \label{ssec:offline_design}
In this section, we show how to design the terminal cost and then summarize the offline and online computations of the proposed MPC formulation. 
In contrast to standard economic MPC designs, we do not enforce a terminal set constraint that ensures that the prediction $\hat{x}_{N|k}$ is in a small neighborhood of the steady-state $\hat{x}_k^{\mathrm{s}}$~\cite[Assumption~6]{amritEconomicOptimizationUsing2011}. Instead, we design a terminal cost that is valid on the full constraint set.
We further denote $\kappa_{\mathrm{s}}(v,\theta) := \kappa(h(v,\theta), h(v,\bar{\theta}), v)$ and $\ell^\mathrm{s}(v,\theta) := \ell(h(v,\theta),\ \kappa_\mathrm{s}(v,\theta))$ to simplify the notation.

\begin{assumption}[Terminal cost] \label{assump_terminal-ingredients}
    There exists a uniform constant $L_\mathrm{f}\geq0$ such that for any $x,\tilde{x},x^\mathrm{s},\tilde{x}^\mathrm{s}\in\mathbb{Z}_x,\ v,\tilde{v}\in\mathbb{Z}_u$, and any $\theta, \tilde{\theta} \in \Theta$, it holds that
    \begin{subequations}
    \label{assumpt_terminal-cost-eqs}
    \begin{align}
        \label{assump_terminal-cost-decrease-eq}
        \begin{split}
           &\hspace{-2cm}\ell_\mathrm{f}(f(x,\kappa(x,z^\mathrm{s},v),\theta),\  \hat{x}^\mathrm{s},\ v,\ \theta) - \ell_\mathrm{f}(x,\ \hat{x}^\mathrm{s},\ v,\ \theta) \\
            &\bquad\bquad\, \leq -\ell(x,\, \kappa(x, z^\mathrm{s}, v))+ \ell^\mathrm{s}(v,\theta),
        \end{split}\\[5pt]
        \begin{split}
        \|\ell_{\mathrm{f}}(x,x^\mathrm{s},v,\theta) &- \ell_{\mathrm{f}}(\tilde{x}, \tilde{x}^\mathrm{s}, \tilde{v}, \tilde{\theta})\| \\
        &\hspace{-2cm}\leq L_\mathrm{f} \left( \!\|x- \tilde{x}\|  + \|x^\mathrm{s}-\tilde{x}^\mathrm{s}\| + \|v-\tilde{v}\| + \| \theta-\tilde{\theta} \|\!\right),
        \end{split} \label{assump_Lf}
    \end{align}
    \end{subequations}
    with $z^\mathrm{s} = h(v,\bar{\theta})$ and $\hat{x}^\mathrm{s} = h(v,\theta)$.
\end{assumption}

Inequality~\eqref{assump_terminal-cost-decrease-eq} corresponds to the standard inequality posed on terminal costs in economic MPC (cf.~\cite[Assumption~6]{amritEconomicOptimizationUsing2011}), but needs to be valid for all parameters $\theta\in\Theta$ and setpoints $x^\mathrm{s}$ to uncertain parameters.
Furthermore, $\ell_\mathrm{f}$ is Lipschitz-continuous~\eqref{assump_Lf}.

\subsubsection*{Constructing $\ell_\mathrm{f}$}
Next, we we show how to construct a terminal cost $\ell_\mathrm{f}$ that satisfies Assumption~\ref{assump_terminal-ingredients}.
In particular, we use a linear-quadratic terminal cost
\begin{align} \label{lin-quad_terminal-cost}
       \ell_\mathrm{f}(x,\hat{x}^\mathrm{s}, v^\mathrm{s}, \theta) &:= \|x-\hat{x}^\mathrm{s}\|^2_{P_\mathrm{f}} + p(\theta, \hat{x}^\mathrm{s}, v^\mathrm{s})^\top (x-\hat{x}^\mathrm{s}),
    \end{align}
with $P_\mathrm{f}\succ 0$. The matrix $P_{\mathrm{f}}$ is computed offline, whereas the linear term $p(\theta, \hat{x}^\mathrm{s}, v^\mathrm{s})$ has an analytical formula and can be included in Problem~\eqref{MPC_scheme}. 

In order to compute $P_\mathrm{f}$ and $p$ such that $\ell_{\mathrm{f}}$ satisfies Assumption~\ref{assump_terminal-ingredients}, we introduce four quantities:
First, we define the \emph{shifted cost} as 
\begin{align} \label{shifted_stage_cost}
    \bar{\ell}(x, v^\mathrm{s}, \theta) := \ell(x,\kappa(x,h(v^\mathrm{s},\bar{\theta}),v^\mathrm{s})) - \ell^\mathrm{s}(v^\mathrm{s},\theta).
\end{align}
Second, the Jacobian of the closed-loop system $f_\kappa (x,v,\theta) := f(x,\kappa(x,h(v,\bar{\theta}),v),\theta)$ is given by
    \begin{align} \begin{split}\label{offline_linearization_AK-def-eq}
        \bquad A_K&(x,v, {\theta}) := \left.\frac{\partial f_\kappa}{\partial x}\right|_{(x,v,\theta)}
    \end{split}
    \end{align}
Third, based on the shifted cost $\bar{\ell}$ and the Jacobian $A_K$, we can define the vector $p$ as
\begin{alignat}{2}
        \begin{split}
        &p(\theta, \hat{x}^s, v^\mathrm{s})  := \left[\left(I_{n_x} -  A_K(\hat{x}^s, v^\mathrm{s},{\theta})\right)^{-1}\right]^\top \cdot  \left.\frac{\partial \bar{\ell}}{\partial x}\right|^\top_{(\hat{x}^\mathrm{s}, v^\mathrm{s},\theta)}.\label{numex_p-vector-def}
        \end{split}
    \end{alignat}
Fourth, we use $H\geq0$ to denote an upper bound on the maximum eigenvalue of the Hessian of the closed-loop dynamics $\frac{\partial^2 [f_\kappa]_i}{\partial x^2}$, i.e., 
\begin{align}
    H &:= \max_{\substack{x\in\mathbb{Z}_x,\ v^\mathrm{s}\in\mathbb{Z}_u,\\ \theta \in \Theta,\ i\in\mathbb{I}_{[1,n_x]}}} \left\{
                \lambda_{\max} \left( \left.\frac{\partial^2 [f_\kappa]_i}{\partial x^2}\right|_{(x,v^\mathrm{s}, \theta)} \right)
            ;\ 0 \right\}.
        \label{term-cost_H_def-eq}
\end{align} 
The following proposition computes a matrix $P_\mathrm{f}$ such that the terminal cost~\eqref{lin-quad_terminal-cost} satisfies Assumption~\ref{assump_terminal-ingredients}.

\begin{proposition}[Terminal cost]          \label{prop_constr-term-ingred}
     Let Assumptions~\ref{assump_bounded-disturb-param}, \ref{assump_regularity}, and~\ref{assump_equi-manifold} hold and  
     suppose that the dynamics $f_\mathrm{w}$, the stage cost $\ell$, and the feedback $\kappa$ are twice continuously differentiable.     
    Consider any $Q\succeq 0$ and $\alpha\geq0$ satisfying
    \begin{subequations}\label{prop_term-cost-eqs}
         \begin{flalign} 
            Q &\succeq \left.\frac{\partial^2\bar{\ell}}{\partial x^2}\right|_{(x,v^\mathrm{s},\theta)} + \alpha I_{n_x},\quad \forall x\in\mathbb{Z}_x, v^\mathrm{s}\in \mathbb{Z}_u, \theta\in\Theta,  \label{Q_def-eq}\\
            \alpha &\geq  \sqrt{n_x}\cdot \frac{H}{2} \cdot \max_{\hat{x}^\mathrm{s},v^\mathrm{s}\in\mathbb{Z}_x, \theta\in\Theta} \|p(\theta, \hat{x}^s, v^\mathrm{s})\|, \label{alpha_def-eq}
            \\
            \text{and } P_\mathrm{f}&\succ 0 \text{ satisfying}  \nonumber\\
             A_K&(x,v^\mathrm{s},\theta)^\top P_\mathrm{f} A_K(x,v^\mathrm{s},\theta) - P_\mathrm{f} \preceq -Q,  \nonumber\\
             &\qquad \qquad \qquad \forall x\in\mathbb{Z}_x,\ v^\mathrm{s}\in\mathbb{Z}_u,\  \theta\in\Theta.  \label{numex_Lyapunov_inequality}
         \end{flalign}
    \end{subequations}
    Then, the terminal cost~\eqref{lin-quad_terminal-cost} satisfies Assumption~\ref{assump_terminal-ingredients}.
    
\end{proposition}
\begin{proof}
    The proof is provided in Appendix~\ref{app:proofs}.
\end{proof}

Compared to the standard design procedure for the economic terminal cost $\ell_{\mathrm{f}}$ from~\cite{amritEconomicOptimizationUsing2011}, we provide a terminal cost that depends on $\theta$ and $\hat{x}^\mathrm{s}$. More importantly, we guarantee that $\ell_\mathrm{f}$ satisfies~\eqref{assumpt_terminal-cost-eqs} for all $x\in\mathbb{Z}_x$, not just in a small enough terminal set. 

To obtain the terminal cost, the constant $\alpha$ needs to be computed to obtain the matrix $Q$. Thereby, we compute a global upper bound on the Hessian of the shifted cost. Solving the linear matrix inequality~\eqref{numex_Lyapunov_inequality} yields the matrix $P_\mathrm{f}$. The vector $p$ must be evaluated online.

\begin{remark}[State-dependent matrices $P_\mathrm{f}(x)$]
    Proposition~\ref{prop_constr-term-ingred} considers a constant matrix $P_\mathrm{f}$ for simplicity of exposition. 
    We conjecture that a state-dependent matrix $P_\mathrm{f}(x)$ can also be used, similar to~\cite[Lemma~1]{kohlerNonlinearModelPredictive2020},~\cite[Corollary~2]{kohlerPeriodicOptimalControl2020}.
\end{remark}

\begin{remark}[Terminal cost] \label{rem_termin_ingred}
Appendix~\ref{ssec:app_SDP} shows how to constructively compute the terminal cost $\ell_\mathrm{f}$. The corresponding semi-automated code is provided online: \\\texttt{\url{https://doi.org/10.3929/ethz-c-000792397}}. \\Choosing a larger (conservative) terminal cost will yield a worse transient performance, similar to the effect of a larger terminal costs in any standard MPC implementation~\cite[Thm.~5.22]{grune2017book}.
Practically, the terminal cost $\ell_{\mathrm{f}}$ can be approximated by using an extended horizon where a fixed stabilizing feedback $\kappa$ is simulated~\cite{liuEconomicModelPredictive2016}.
More broadly, the relevance of the terminal ingredients for MPC in industrial implementations is discussed in~\cite{mayneApologiaStabilisingTerminal2013}.
\end{remark}

A summary of the offline design and the online operation is provided by Algorithm~\ref{MPC_algorithm}. Compared to a robust economic MPC scheme~\cite{bayerTubebasedRobustEconomic2014}, the online solution of the adaptive economic MPC problem~\eqref{MPC_scheme} additionally requires the LMS estimator~\eqref{LMS_update-eqs}, the prediction of the certainty-equivalent trajectory~\eqref{MPC_scheme-d}, and the optimization over the steady-state $\hat{x}^{\mathrm{s}}_k$ (as the optimal steady-state is not known beforehand). The added computational and design complexity tends to be negligible. 
\begin{algorithm}[ht] 
\caption{Robust Adaptive Economic MPC}\label{MPC_algorithm}
\begin{algorithmic}[0]
\State Compute $\mu$ according to~\eqref{probsetup_LMS_gain-definition}.
\State Choose a tube scheme (i.e., $\Phi$ and $\mathbb{X}_{\cdot|k}$) and corresponding auxiliary feedback $\kappa$ in agreement with Assumption~\ref{assump_Phi-function}.
\State Compute terminal cost $\ell_{\mathrm{f}}$ 
(e.g., Proposition~\ref{prop_constr-term-ingred}).
\State Set $\lambda_0=\infty$.
\For{$k\in\mathbb{N}$}
    \State Measure the state $x_k$.
    \State Update $\hat{\theta}_k \in \Theta$ using the LMS estimator~\eqref{LMS_update-eqs}.
    \State Update $\lambda_k$ using~\eqref{eq:lambda_k}.
    \State Solve the optimization problem~\eqref{MPC_scheme}.
    \State Apply the control input~\eqref{MPC-scheme_optimal-input}.
\EndFor
\end{algorithmic}
\end{algorithm}

\section{Theoretical Analysis} \label{sec:theo_analysis}
First, we present our main result, a general transient performance bound (Section~\ref{ssec:transient_perf}). 
Subsequently, we consider disturbances with finite energy and parameter variations with finite path length to derive intuitive bounds on the asymptotic average performance (Section~\ref{ssec:asymp_performance}), including a suboptimality bound (Section~\ref{ssec:subopt_guarantee}).
Lastly, we relate the obtained results to the state-of-the art (Section~\ref{ssec:discuss}).

\subsection{Transient Performance} \label{ssec:transient_perf}
We are now ready to state the main result, the transient performance guarantee for the closed-loop system.

\begin{theorem}[Transient performance bound] \label{thm_transient}
    Let Assumptions~\ref{assump_bounded-disturb-param}--\ref{assump_terminal-ingredients} hold and suppose that Problem~\eqref{MPC_scheme} is feasible at time $k=0$. Then, Problem~\eqref{MPC_scheme} is recursively feasible and the constraints~\eqref{prob_setup:constraints-eq} are robustly satisfied for the closed-loop system resulting from Algorithm~\ref{MPC_algorithm}. Furthermore, there exist uniform constants $C_\mathcal{J}, C_A, L_\mathrm{s}\geq 0$ such that, for all $T\in\mathbb{N}$, the closed-loop cost satisfies
    \begin{align}\label{thm_transient-eq}
        &\sum_{k=0}^{T-1}\left[\ell(x_k,u_k)-\ell(\hat{x}^\mathrm{s\ast}_k, \hat{u}^\mathrm{s\ast}_k)\right] \nonumber \\
        & \leq (C_A + \sqrt{\mu}\beta L_\mathrm{s})\cdot \left(2\sqrt{T} |w_{[0,T-1]}|_{\mathcal{L}_2} \vphantom{\sqrt{\frac{T}{\mu}}} \right.\left.+\sqrt{\frac{T}{\mu}}\|\hat{\theta}_0-\theta_0\| \right.\nonumber\\
        &\quad \left. + \sqrt{\frac{2Tc_\theta}{\mu}}\sum_{k=0}^{T-1}  \sqrt{ \|\Delta\theta_k\|}\right) + C_\mathcal{J}.
    \end{align}
\end{theorem}

\begin{proof}
    The proof is provided in Appendix~\ref{app:proofs}.
\end{proof}
Theorem~\ref{thm_transient} relates the transient performance of the closed-loop system from Algorithm~\ref{MPC_algorithm} to the performance at a certainty-equivalent, artificial steady-state $(\hat{x}^\mathrm{s}_k,\hat{u}^\mathrm{s}_k)$.

Inequality~\eqref{thm_transient-eq} is a qualitative performance bound, which scales affinely with the $\mathcal{L}_2$-norm of the realized disturbances $|w_{[0,T-1]}|_{\mathcal{L}_2}$, affinely with the initial parameter estimation error $\|\hat{\theta}_0-\theta_0\|$, and with the variation of the true parameters $\Delta \theta_k$. 
The constant $C_\mathcal{J}$ bounds the optimal cost $\mathcal{J}_M^\ast(x_0,\theta_0,\lambda_0)$ at initialization, see~\eqref{eq:C_J-definition}. 
The constants $L_\mathrm{s}$ and $C_A$ depend on the continuity of the steady-state cost $\ell^\mathrm{s}$, cf.~\eqref{transient-proof_eq-c}, and the continuity of the dynamics (see Prop.~\ref{prop_consecutive-distances}), respectively.
While large weights $\beta$ deteriorate the transient performance, they improve the asymptotic average cost (cf. Cor.~\ref{corr_subopt}). Thus, choosing $\beta$ is a trade-off between the long-term asymptotic and the short-term transient performance.
Similarly, if the terminal cost $\ell_\mathrm{f}$ satisfies~\eqref{assump_Lf} only up to a tolerance $\varepsilon_\mathrm{f}>0$ (cf. Remark~\ref{rem_termin_ingred}), the transient performance also scales with this tolerance.

Assumption~\ref{assump_bounded-disturb-param} implies that $|w_{[0,T-1]}|_{\mathcal{L}_2}\leq \hat{w}\sqrt{T}$ with $\hat{w}=\max_{d\in\mathbb{D},(x,u)\in\mathbb{Z}} \|E(x,u)d\|$. Thus, the asymptotic average performance scales linearly with $\hat{w}$.

\subsection{Asymptotic Average Performance} \label{ssec:asymp_performance}
In the following, we derive a more intuitive asymptotic performance bound for the special case of finite-energy disturbances $w_k$ and finite path length variations of the parameters $\theta_k$. 
\begin{assumption}[Finite-energy disturbances, finite parameter path length]\label{assump_finite-energy}
    There exist constants $S_{\mathrm{w}},\,S_\theta <\infty$ such that
    \begin{subequations}
    \begin{align} 
        \lim_{T\rightarrow\infty} |w_{[0,T]}|_{\mathcal{L}_2}^2 = \lim_{T\rightarrow\infty} \sum_{k=0}^T \|w_k\|^2 &\leq S_{\mathrm{w}}, \label{assump_finite-energy_parameter-eq-a}\\
        \lim_{T\rightarrow \infty}\sum_{k=0}^T \|\theta_{k+1}-\theta_k\| &\leq S_\theta. \label{assump_finite-energy_parameter-eq-b}
    \end{align} \end{subequations}
\end{assumption}
This assumption is stronger than assuming bounded disturbances in Assumption~\ref{assump_bounded-disturb-param} but it allows us to relate our asymptotic average performance to existing results that assume zero model mismatch. In particular, Assumption~\ref{assump_finite-energy} implies that the averaged prediction error tends to zero and can be used to study the convergence of the artificial steady-state.
\begin{lemma}[Asymptotic average of one-step prediction error]\label{lemma_asymp-lin-bound}
    Let Assumptions~\ref{assump_bounded-disturb-param} and~\ref{assump_finite-energy} hold and let $(x_k,u_k)\in \mathbb{Z}$ for all $k\in\mathbb{N}$. Then, the asymptotic average prediction error is zero, i.e.,
    \begin{align} \label{lemma_asymp-lin-bound_equation}
         \lim_{T \rightarrow \infty} \sum_{k=0}^{T-1} \frac{ \|\tilde{x}_{1|k}\| + \|w_k\|}{T} =0.
    \end{align}
\end{lemma}
\begin{proof}
    The proof is provided in Appendix~\ref{appendix_aux-results}.
\end{proof}
Using Lemma~\ref{lemma_asymp-lin-bound}, we can bound the asymptotic average performance of the closed-loop system.
\begin{corollary}[Asymptotic average performance]\label{corr_asymptotic_performance}
    Suppose the conditions in Theorem~\ref{thm_transient} are fulfilled and also Assumption~\ref{assump_finite-energy} holds. Then, the closed-loop performance satisfies
    \begin{align} \label{corr_asymptotic_performance-eq}
        &\limsup_{T\rightarrow\infty} \sum_{k=0}^{T-1} \frac{  \ell(\hat{x}_k,\hat{u}_k)  }{T} \leq \lambda_\infty,
    \end{align}
    where $\displaystyle \lambda_\infty:=\limsup_{k\rightarrow\infty}\lambda_k = \limsup_{k\rightarrow\infty}\ell(\hat{x}^\mathrm{s\ast}_k,\hat{u}^\mathrm{s\ast}_k)$.
\end{corollary}

\begin{proof}
The proof is provided in Appendix~\ref{app:proofs}.
\end{proof}
This corollary shows that, on average, the closed-loop performance is not worse than the performance at the optimized steady-state $\hat{x}^{\mathrm{s\ast}}$.

\subsection{Suboptimality Guarantee} \label{ssec:subopt_guarantee}
In the following, we strengthen Corollary~\ref{corr_asymptotic_performance} by studying the convergence of the artificial references and thereby providing an asymptotic performance guarantee with respect to the best achievable equilibrium. 
The optimal achievable equilibrium cost $\underline{\ell}$ is defined as
\begin{align}
\begin{split}
     \underline{\ell}(x,\theta, \lambda)  & := \min_{v \in \mathbb{V}_\mathrm{s}(x,\theta, \lambda)} \ell^\mathrm{s}(v,\theta), \label{def_opt_achiev_cost}\\[3pt]
     \mathbb{V}_\mathrm{s}(x,\theta,\lambda) &= \{ v^\mathrm{s}:\ v^\mathrm{s} \text{ feasible in~\eqref{MPC_scheme}} \\
     & \qquad \text{with } x_k=x,\hat{\theta}_k=\theta, \lambda_k=\lambda \}.
\end{split} 
\end{align}
Condition~\eqref{def_opt_achiev_cost} yields the smallest cost that the artificial setpoint 
$(\hat{x}^{\mathrm{s}},\hat{u}^{\mathrm{s}})$ can attain for any feasible solution of Problem~\eqref{MPC_scheme}.  
The optimal achievable equilibrium cost $\underline{\ell}$ exists due to the compact sets $\Theta,\, \mathbb{Z}$ and continuity of $\ell^{\mathrm{s}}$. However, we do not assume that the corresponding optimal steady-state is unique. 
In~\cite{fagianoGeneralizedTerminalState2013}, it was shown that we can approximately achieve this level of performance by increasing the weighting $\beta$ of the artificial setpoint term in~\eqref{MPC_scheme-a}. 
The following corollary extends this result to the considered adaptive economic MPC formulation. 
\begin{corollary}[Approximate optimality]
\label{corr_subopt}
    Suppose the conditions in Corollary~\ref{corr_asymptotic_performance} are fulfilled. 
    Then, there exists a function $\underline{\beta}(\varepsilon):\,\mathbb{R}_{\geq 0}\rightarrow\mathbb{R}_{\geq0}$, such that for any $\varepsilon> 0$, choosing $\beta\geq \underline{\beta}(\varepsilon)$ ensures 
\begin{align}
    \begin{split}
        \limsup_{T\rightarrow\infty} \sum_{k=0}^{T-1} \frac{  \ell(x_k,u_k)  }{T} 
        &\leq 
        \limsup_{k\rightarrow\infty} \ell(\hat{x}^\mathrm{s}_k, \hat{u}^\mathrm{s}_k) \label{corr_subopt-eq-b}\\
        &\leq \limsup_{k\rightarrow\infty} \underline{\ell}(x_k,\hat{\theta}_k,\lambda_k) +\varepsilon. 
    \end{split}
\end{align}
\end{corollary}
\begin{proof}
    The proof is provided in Appendix~\ref{app:proofs}.
\end{proof}

Corollary~\ref{corr_subopt} guarantees that the average performance of the closed-loop system to be at least $\varepsilon$-close to the best achievable equilibrium cost. The proof shows that $\lambda_{k+1} \leq  \lambda_k+L_\mathrm{s}\|\Delta \hat{\theta}_k\|$ (cf.~\eqref{lambda_decrease_eq}), i.e., the economic cost at the artificial setpoints $(\hat{x}^\mathrm{s}_k,\hat{u}^\mathrm{s}_k)$ monotonically decreases modulo changes in the parameter estimate, which are asymptotically vanishing due to Lemma~\ref{lemma_asymp-lin-bound}. By choosing $\beta$ sufficiently large, the guaranteed asymptotic average performance is \emph{arbitrarily close} to the best achievable steady-state performance.

\subsection{Discussion} \label{ssec:discuss}
In the following, we discuss the derived performance bounds and relate our results to the state-of-the-art for economic and adaptive MPC formulations. 
\paragraph{Performance bounds}
Theorem~\ref{thm_transient} provides a general transient bound on the closed-loop performance when applying the proposed adaptive economic MPC scheme to uncertain nonlinear systems. 
In particular, it bounds the suboptimality of the closed-loop operation relative to operating the system at the optimized setpoint $\hat{x}^{\mathrm{s\ast}}$. This bound depends on the initial error in the parameter estimate, the energy of the disturbances, and the variation of the unknown parameters. 
Notably, our analysis does not assume persistently exciting system behavior or convergence of the LMS estimates to the true parameters, which would limit the applicability. 
If we consider long-term operation $T\rightarrow\infty$ and bounded variations in the unknown parameters, then the bound on the average scales linearly with $\lim_{T\rightarrow\infty}\dfrac{|w_{[0,T-1]}|_{\mathcal{L}_2}}{\sqrt{T}}\leq \max_{k\in\mathbb{I}_{[0,T]}}\|w_k\|$. 
Thus, the asymptotic average closed-loop cost is smaller or equal to the cost at the optimized steady-state if the disturbances decay, independent of the initial error in the parameter estimate. 
This is formalized in Corollaries~\ref{corr_asymptotic_performance}-\ref{corr_subopt}, ensuring that the average cost is at least $\varepsilon$-close to the best achievable setpoint if we use a large penalty $\beta$ and the disturbance energy and parameter variations are finite (Assumption~\ref{assump_finite-energy}). 
Notably, in general the closed-loop system may also outperform operation at the (unknown) optimal steady-state, see also the numerical example in Section~\ref{sec:num_ex}. 
According to~\eqref{thm_transient-eq}, a large penalty $\beta$ may also increase the impact of disturbances and parameter variations and hence, a trade-off should be considered when choosing $\beta$ in practical applications. 
\paragraph{Adaptive MPC}
Adaptive MPC schemes using LMS adaptation are investigated in~\cite{lorenzenRobustMPCRecursive2019,kohlerLinearRobustAdaptive2019,kohlerRobustAdaptiveModel2021}. 
These formulations all address the problem of stabilizing a fixed steady-state and ensure finite-gain $\mathcal{L}_2$-stability, i.e., a linear bound between the tracking error and the disturbance energy. Notably, all of these designs assume that some optimal setpoint satisfies the dynamics for all possible parameters $\theta\in\Theta$. This can be very restrictive. 
The proposed design removes this limitation by using artificial references and a generalized terminal constraint in the design. 
\paragraph{Online optimized setpoints in Economic MPC}
The proposed approach optimizes an artificial reference to determine the optimal operating regime online, while also accounting for the updated model parameters. 
Existing results for economic MPC with artificial references can be found in~\cite{fagianoGeneralizedTerminalState2013,mullerEconomicModelPredictive2013,mullerPerformanceEconomicModel2014,kohlerPeriodicOptimalControl2020}. 
These results all study the nominal case, i.e., there are no disturbances and the parameters $\theta$ are known and constant. 
In~\cite[Theorem~3]{mullerPerformanceEconomicModel2014}, it is proposed to increase the weighting $\beta$ during online operation, ensuring that the asymptotic average performance is no worse than operation at the optimal achievable steady-state. 
Closer to the proposed approach, \cite{fagianoGeneralizedTerminalState2013} shows that the asymptotic average closed-loop performance is almost as good as the performance at the optimal achievable steady-state by choosing a large constant weight $\beta$. 
In particular, Corollary~\ref{corr_subopt} generalizes~\cite[Theorem~2]{fagianoGeneralizedTerminalState2013}, ensuring the same performance bound but for uncertain systems subject to disturbances and time-varying uncertain parameters. This is made possible by including a robust tube and online model adaptation in the design.
\paragraph{Terminal ingredients}
Another important contribution is the design of the terminal cost in Proposition~\ref{prop_constr-term-ingred}. 
While designs of economic terminal costs are available~\cite{amritEconomicOptimizationUsing2011}, the application of these designs to online optimized setpoints is non-trivial~\cite[Sec.~4.1]{mullerPerformanceEconomicModel2014}, let alone to uncertain system with online adaptation. 
In particular, existing designs use a local Taylor expansion around a fixed setpoint to ensure Inequality~\eqref{assumpt_terminal-cost-eqs} holds in a sufficiently small neighborhood. 
However, due to the uncertainty in the system dynamics~\eqref{prob_setup:eq-dynamics}, it is not feasible to force all uncertain predictions to end in a small terminal set around a fixed equilibrium. 
Proposition~\ref{prop_constr-term-ingred} addresses this issue by providing a design that is valid for any setpoint $x^{\mathrm{s}}$, any model parameters $\hat{\theta}$, and for all states $x\in\mathbb{Z}_x$, without requiring the explicit restriction to a small terminal set.
In Appendix~\ref{ssec:app_SDP}, we show how to construct a terminal cost by solving a semidefinite program (cf. Remark~\ref{rem_termin_ingred}). 
\paragraph{Prior work}
In our conference paper~\cite{degnerAdaptiveEconomicModel2024}, we derived similar performance guarantees for an adaptive economic MPC scheme. 
This approach was restricted to open-loop stable linear systems with a linear-quadratic economic cost function and soft-constraints. 
In contrast, Algorithm~\ref{MPC_algorithm} ensures robust constraint satisfaction, is applicable to time-varying parameters, and a broad class of nonlinear systems and nonlinear economic costs $\ell$. 
More importantly, the proposed method explicitly deals with the fact that the optimal steady-states $\hat{x}^s_k$ depend on the parameter estimates $\hat{\theta}_k$, while~\cite{degnerAdaptiveEconomicModel2024} assumes that the optimal steady-state is independent of the parameter estimates $\hat{\theta}_k$.

\section{Numerical Example} \label{sec:num_ex}
In this section, we demonstrate improved performance of the proposed adaptive economic MPC (AE-MPC) scheme compared to economic MPC (E-MPC) with perfect parameter knowledge and highlight reliable performance under time-varying parameters.  

The simulations are performed using MATLAB: The offline computations use YALMIP~\cite{YALMIP} and MOSEK~\cite{mosek}, and the online computations use CasADi~\cite{casadi} and IPOPT~\cite{IPOPT}. The code is available online.\footnote{Git: \href{https://github.com/maximiliandegner/AEMPC}{https://github.com/maximiliandegner/AEMPC.git}\\
DOI: \href{https://doi.org/10.3929/ethz-c-000792397}{https://doi.org/10.3929/ethz-c-000792397}}

\subsection{System description}\label{ssec:num_system}
We consider a continuous stirred-tank reactor (CSTR) model from~\cite{baileyCyclicOperationReaction1971}, which is a standard benchmark for economic MPC~\cite{faulwasserEconomicNonlinearModel2018,kohlerPeriodicOptimalControl2020}.
The continuous-time dynamics are given by
\begin{align}\label{numex_sys-desc-dynamics}
    &\frac{\mathrm{d}x}{\mathrm{d}t} = \begin{bmatrix}
        1-[x]_1\\
        -[x]_2 \\
        -[x]_3 + u 
    \end{bmatrix} \\
    &+\begin{bmatrix}
        -[x]_1^2 \exp(-\frac{1}{[x]_3})    & - [x]_1 \exp(-\frac{\delta}{[x]_3}) \\
        [x]_1^2 \exp(-\frac{1}{[x]_3})   & 0\\
        0                                                       & 0
    \end{bmatrix} \!\cdot\! \begin{bmatrix}
        [\theta]_1\cdot 10^5\\
        [\theta]_2\cdot 400
    \end{bmatrix}
    \nonumber
\end{align}
where $\delta=0.55$. The states $x\in\mathbb{R}^3$ are given by the concentration of reactant, the concentration of the desired product, and the temperature of the mixture. The input $u\in\mathbb{R}$ represents the heat flux through the cooling jacket of the reactor. The parameter set is $\Theta = [0.985, 1.015]\cdot [\bar{\theta}]_1 \times [0.985, 1.015]\cdot [\bar{\theta}]_2$ with the nominal parameters $\bar{\theta} = [0.0995, 1.0050]^\top$.
The state and input constraints are given by 
\begin{align} \label{numex_constraints}
    [x]_i &\in [0.03, 0.25] \quad \forall i\in\mathbb{I}_{[1,3]},& u\in [0.049, 0.449].
\end{align}

We use the Euler discretization with a sampling step of $T_s = 0.025$. We consider  additive disturbances $d$, i.e., $E=I_{n_x}$, that are sampled uniformly with $d_k\in \mathbb{D}=5\cdot 10^{-4}\cdot [-1, 1]^3$.
The economic cost function is $\ell(x,u)=-[x]_2$, which corresponds to the maximization of the yields in the desired product through maximizing its concentration $[x]_2$. 
Notably, the optimal operation of this CSTR includes large (quasi-periodic) oscillations, while operation at the optimal steady-state is suboptimal~\cite[Section~3.4]{faulwasserEconomicNonlinearModel2018}. 

\subsection{Offline design} \label{ssec:tube_imple}
The tube propagation (Assumption~\ref{assump_Phi-function}) is implemented with the homothetic tube formulation from~\cite{kohlerComputationallyEfficientRobust2021}, i.e., we parametrize the tube as $\mathbb{X}_{j|k}=\{x\in\mathbb{R}^{n_x}:~\|x-z_{j|k}\|_P\leq s_{j|k}\}$
where $s_{j|k}$ is an online propagated scaling of the tube. 
The auxiliary feedback law takes the form $\kappa(x,z,v) = K(x-z)+v$.
We compute the matrices $P,\ K$, and a contraction rate $\rho$ offline using linear matrix inequalities, see Appendix~\ref{ssec:app_SDP} for details.
The propagation of the tube size takes the parametric uncertainty into account using
\begin{align}
    \label{numex_tube-propagation}
    s_{j+1|k} &= (\rho+L_\mathrm{w}) \cdot s_{j|k} \\
    &\quad +  \max_{\substack{\theta\in\Theta,\\d\in\mathbb{D}}}  \|G(z_{j|k})\cdot\theta +E d\|_P,  
\end{align}
with a Lipschitz bound $L_\mathrm{w}$ given by
\begin{align}
    \bquad L_\mathrm{w} := \max_{\substack{(x,u)\in \mathbb{Z},\\ \theta\in\Theta}} \left\|P^{\frac{1}{2}} \sum_{i=0}^{n_\theta} \left.\frac{\partial [G]_{:,i}}{\partial x}\right|_{x} \cdot [\theta-\bar{\theta}]_i\cdot P^{-\frac{1}{2}}\right\|.
    \label{numex_Lw-definition}
\end{align}
The invariance condition~\eqref{robust-MPC_scheme-f} reduces to $(1-\rho-L_{\mathrm{w}})s_{N|k}\geq  \max_{\substack{\theta\in\Theta,d\in\mathbb{D}}}  \|G(z_{N|k})\cdot\theta +E(z_{N|k})\cdot d\|_P$. The set inclusion~\eqref{robust-MPC_scheme-c} reduces to tightening the constraints on the nominal trajectory $z_{j|k}$ proportional to the scaling $s_{j|k}$. 
Additional details on the formulation of the offline computation of this homothetic tube MPC can be found in Appendix~\ref{ssec:app_SDP} and in~\cite{kohlerComputationallyEfficientRobust2021}.

The terminal cost is designed according to Proposition~\ref{prop_constr-term-ingred} and all MPC schemes are implemented with horizon $N=25$.
The offline computations were performed on a ThinkPad T14 Gen2i (Intel Core i5-1135G7, 16GB RAM) running Windows~11 and took $0.9$ seconds in total.

\subsection{Simulation results} \label{ssec:num_results}
First, we consider the case of \emph{time-invariant} parameters. Then, we show that for \emph{time-varying} parameters, the AE-MPC's performance approaches that of economic MPC with perfect model knowledge.

\paragraph*{Time-invariant parameters} We consider $\theta_k = [0.1, 1]^\top\ \forall k\geq 0$ and the initial parameter estimate $\hat{\theta}_0 = [0.0980, 1.0200]^\top$. We compare three different controllers.
\begin{itemize}
    \item The proposed AE-MPC scheme, with the LMS switched on at time step $200$,
    \item an E-MPC scheme with $\hat{\theta}_k = \theta_k$ (known parameters), and 
    \item an E-MPC scheme with $\hat{\theta}_k = \hat{\theta}_0,\ \forall k\geq 0$ (no adaptation). 
\end{itemize}
In Figure~\ref{fig:results_performance}, we depict the simulation results.
At the beginning, the AE-MPC and the E-MPC (no adaptation) show the same behavior until the LMS is switched on (vertical grey dashed line). Afterwards, the AE-MPC controlled system achieves a higher average concentration $[x]_2$ and thus, a smaller cost. 
AE-MPC first performs similar to E-MPC (no adaptation) and then the performance improves through adaptation and converges to the performance of E-MPC (known parameters).
The parameter estimates $\hat{\theta}_k$ gradually approaches the true parameters, and after $2\cdot 10^3$ time steps, the estimate remains close to the true parameters. 
Furthermore, Theorem~\ref{thm_transient} provides a bound on the transient performance compared to the optimal setpoint, whereas Figure~\ref{fig:results_performance} and Table~\ref{tab:average-perf_comparison} shows that the proposed AE-MPC scheme increases the production yield by $1.038\,\%$ compared to the optimal steady-state.

\begin{figure}
    \centering
    \includegraphics[width=\linewidth]{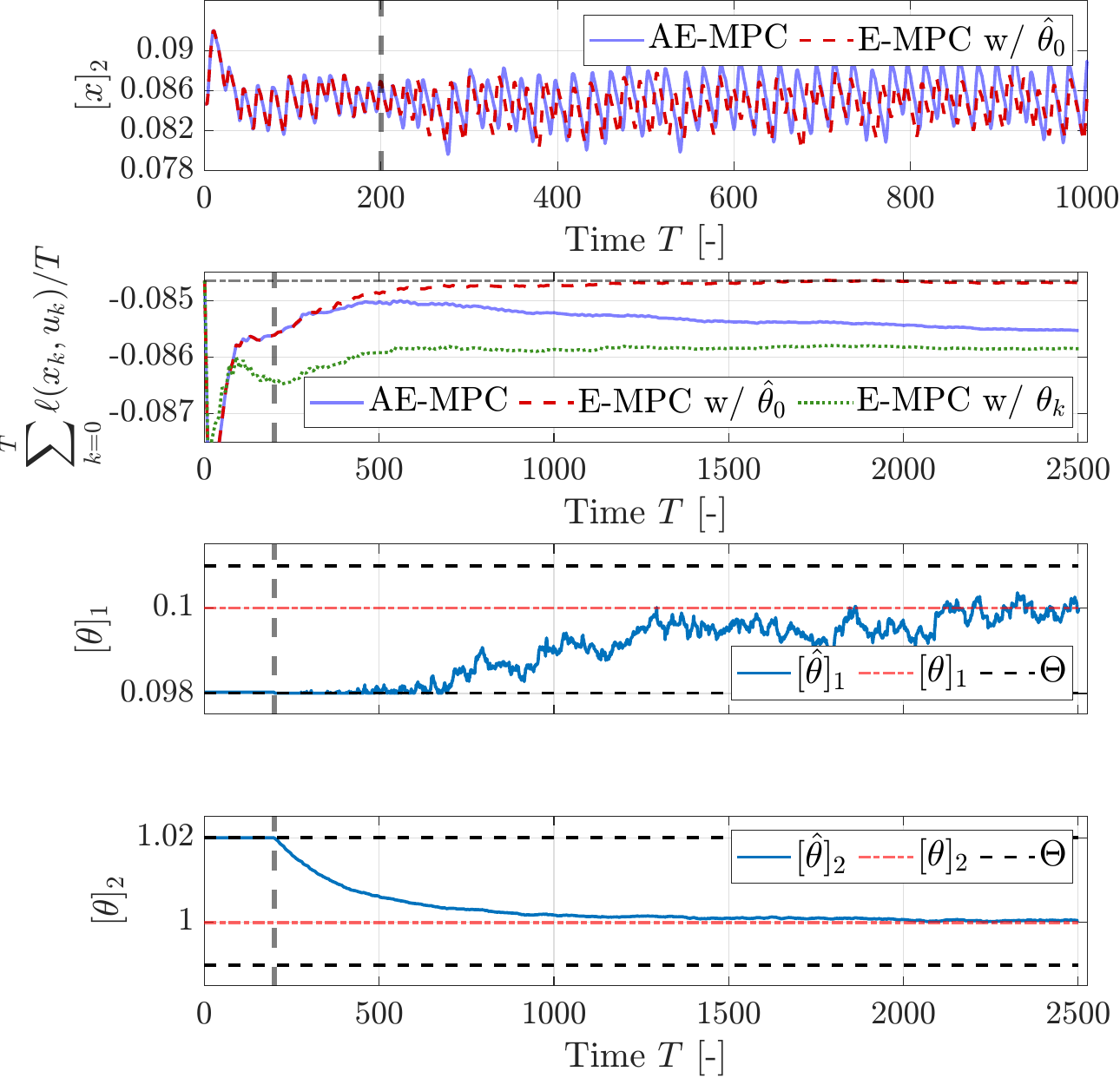}
    \caption{Performance comparison of AE-MPC and E-MPC (no adaptation). First plot: Closed-loop evolution of the state $[x]_2$. Second plot: Averaged transient performance for AE-MPC (blue, solid), E-MPC (no adaptation) with constant parameter estimate $\hat{\theta}_0$ (red, dashed), E-MPC (known parameters) with knowledge of $\theta_k$ (green, dotted), and optimal steady-state cost for $\theta_k$ (grey, dash-dotted). Third and fourth plot: Evolution of the parameter estimates of AE-MPC. The LMS is activated at time step $200$ (vertical grey dashed line).}
    \label{fig:results_performance}
\end{figure}

\begin{figure}
    \centering
    \includegraphics[width=\linewidth]{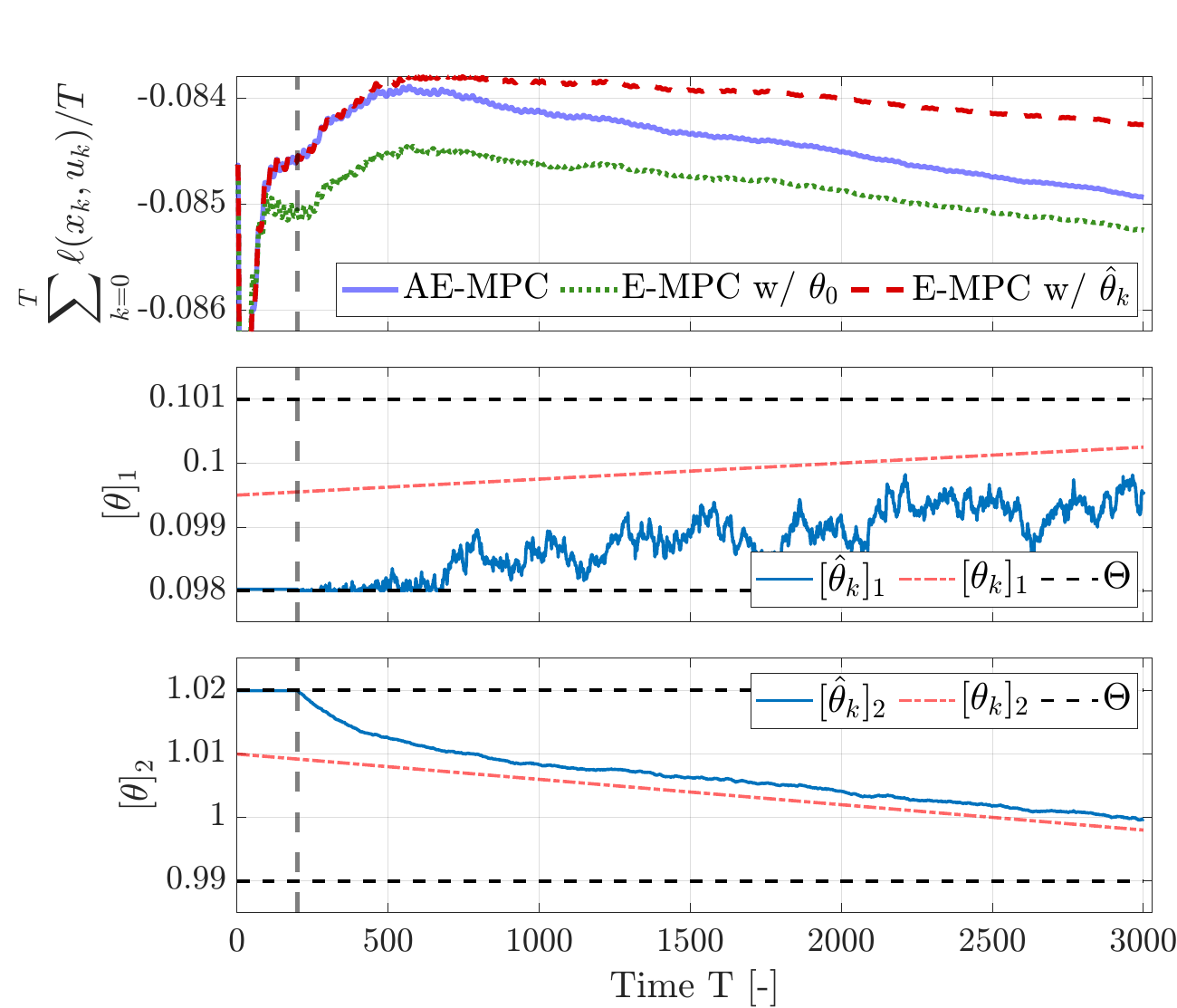}
    \caption{Performance comparison of AE-MPC and E-MPC for time-varying parameters $\theta_k$. First plot: Averaged transient performance for AE-MPC (blue, solid), E-MPC (no adaptation) with constant $\hat{\theta}_0$ (red, dashed), and E-MPC (known parameters) with knowledge of $\theta_k$ (green, dotted). Second and third plot: Evolution of the parameter estimates of AE-MPC. The LMS is activated at time step $200$ (vertical grey dashed line).}
    \label{fig:results_time-varying}
\end{figure}

\begin{table}[h]
\renewcommand{\arraystretch}{1.2}
    \centering
    \begin{tabular}{c c c }
    \toprule
AE-MPC & E-MPC with $\hat{\theta}_0$ & E-MPC with $\theta_k$ \\
     \midrule
        $1.038\%$ &  $0.041\%$ & $1.417\%$ \\
     \bottomrule
    \end{tabular} \vspace{2pt}
    \caption{Average increase in production $(x_2)$ w.r.t. the performance at the (unknown) optimal steady-state for the true parameters.}
    \label{tab:average-perf_comparison}
\end{table}

\paragraph*{Time-varying parameters} In the second simulation, we use a ramp to mimic a drift of the system parameters over $3000$ time steps. The evolution of the LMS parameter estimates in this case are shown in Figure~\ref{fig:results_time-varying} together with the averaged transient performance. The results indicate that the LMS parameter estimator can follow drifting parameters which in turn leads to an improvement of the average performance. After $2.5\cdot10^3$ time steps, the performance of AE-MPC and E-MPC (known parameters) are almost identical and significantly better than the performance of E-MPC without adaptation. 

In summary, the numerical comparison demonstrated that the adaptive economic MPC scheme yields a performance that is significantly better than that of E-MPC (no adaptation) and in fact closely approximates the performance of E-MPC (known parameters). 
\section{Conclusion} \label{sec:concl}
We provided a Model Predictive Control (MPC) scheme for directly minimizing the economic transient cost of uncertain nonlinear systems. 
The proposed MPC formulation utilizes least-mean squares parameter adaptation, minimizes a certainty-equivalent economic cost, employs artificial references to optimize the unknown optimal steady-states online, and uses a robust formulation to ensure constraint satisfaction for all possible parameters and disturbances.
Our theoretical results show that the transient closed-loop cost is bounded by the cost at the optimized setpoint and a penalty that scales linearly with the disturbances and the initial parameter mismatch and increases with the parameter variations. Furthermore, for finite-energy disturbances and finite variations of the parameters, the asymptotic average cost is (approximately) not worse than the cost at the best reachable
steady-state.
Applying the proposed MPC scheme to a continuous stirred-tank reactor in simulation showed that the controller recovered the performance of economic MPC with knowledge of the true parameters for both time-invariant and time-varying parameters.

\section*{Acknowledgments}
The authors are grateful for helpful comments by Matthias A. M\"uller on the proof of Corollary~\ref{corr_asymptotic_performance}.
\appendix
First, we provide the proofs of the main results, i.e., Propositions~\ref{probsetup_prop_LMS-bound} and~\ref{prop_constr-term-ingred}, Theorem~\ref{thm_transient}, and Corollaries~\ref{corr_asymptotic_performance} and~\ref{corr_subopt} (Appendix~\ref{app:proofs}). Second, we provide the proofs of the auxiliary lemmas (Appendix~\ref{appendix_aux-results}). Lastly, we describe the offline computations for the robust MPC used in the numerical example (Appendix~\ref{ssec:app_SDP}).

%%%%%%%%%%%%%%%%%%%%%%%%%%%%%%%%%%%%%%%%%%%%%%%%%%%%%%%%%%%%%%%%%
\subsection{Proofs of the Main Results} \label{app:proofs}

\subsubsection{Proof of Proposition~\ref{probsetup_prop_LMS-bound}}
The proof follows the arguments in~\cite[Lemma 5]{lorenzenRobustMPCRecursive2019} and extends them to time-varying parameters. We abbreviate $G_k := G(x_k,u_k)$. 
Non-expansiveness of the projection~\eqref{LMS_update_b} on the convex set $\Theta$ implies 
\begin{align} 
    \bquad \|\hat{\theta}_{k+1} - \hat{\theta}_k\|  
        % \stackrel{\eqref{prelim_proj_nonexpansiveness}}{\leq} 
        \leq \|\tilde{\theta}_{k+1}-\hat{\theta}_k\| % \nonumber\\
        % &\!\stackrel{\eqref{LMS_update_a}}{\leq}  \mu \|G(x,u) \cdot (\tilde{x}_{1|k}+w_k)\| \label{probsetup_LMS-bound_proof_aux-interim-result} \\
        %\leq 
        \stackrel{\eqref{LMS_update_a}}{\leq}\mu \|G_k\| \cdot \|\tilde{x}_{1|k}+w_k\|.
        %\|G(x_k,u_k)\| \cdot \|\tilde{x}_{1|k}+w_k\|.
        \label{probsetup_LMS-bound_proof_aux-result}
    \end{align}
    Applying~\eqref{probsetup_LMS_gain-definition} yields~\eqref{probsetup_prop_LMS-theta-diff-eq}. \\
To prove~\eqref{probsetup_prop_LMS-bound-eq}, we complete the square and obtain
%relate $\hat{\theta}_{k+1}-\theta_{k+1}$ to $\Delta \theta_k$. Thereby, we obtain
    \begin{align}
        &\|\hat{\theta}_{k+1} - \theta_{k+1}\|^2 
        \nonumber\\
        &\quad = \|\hat{\theta}_{k+1} - \theta_k\|^2 + 2\Delta \theta_k^\top (\theta_{k+1} -\Delta\theta_k - \hat{\theta}_{k+1}) +\|\Delta \theta_k\|^2\nonumber\\
        &\quad = \|\hat{\theta}_{k+1} - \theta_k\|^2 + 2\Delta \theta_k^\top (\theta_{k+1} - \hat{\theta}_{k+1})-\|\Delta \theta_k\|^2\nonumber\\
        &\quad \leq \|\hat{\theta}_{k+1} - \theta_k\|^2 + 2c_\theta \|\Delta \theta_k\|,\label{proof_LMS_intermediate-a}
    \end{align}
 where $c_\theta \geq \| \theta_{k+1}-\hat{\theta}_{k+1}\|$.
 By using the system dynamics~\eqref{prob_setup:eq-dynamics} and the definition of $\hat{x}_{1|k}$, we can rewrite the update equation~\eqref{LMS_update_a} as
 \begin{align*}
     \tilde{\theta}_{k+1} 
    \stackrel{\eqref{prob_setup:eq-dynamics}, \eqref{LMS_one-step_prediction-eq},\eqref{LMS_update_a}}{=}\hat{\theta}_k + \mu G_k^\top \cdot \left(G_k\cdot(\theta_k - \hat{\theta}_k) + w_k\right).
 \end{align*}
 Following the proof of~\cite[Lemma~5]{lorenzenRobustMPCRecursive2019}, we have
    \begin{align}
        &\frac{1}{\mu}\|\tilde{\theta}_{k+1} - \theta_{k}\|^2 - \frac{1}{\mu}\|\hat{\theta}_k - \theta_k\|^2\nonumber\\
        &\stackrel{\eqref{LMS_update_a}}{\leq} \frac{1}{\mu} \left( \|\mu G_k^\top(\tilde{x}_{1|k}+w_k)\|^2 +2(\tilde{x}_{1|k}+w_k)^\top G_k(\hat{\theta}_k-\theta_k) \right)\nonumber\\
        &\ \stackrel{}{=} (\mu\|G_k\|^2-1)\|\tilde{x}_{1|k}+w_k\|^2 - 2(\tilde{x}_{1|k}+w_k)^\top \tilde{x}_{1|k} \nonumber\\
        &\, \stackrel{\eqref{probsetup_LMS_gain-definition}}{\leq} -\|\tilde{x}_{1|k}\|^2 + \|w_k\|^2, \label{proof_LMS_intermediate-b}
    \end{align}
    where the last inequality holds thanks to completing the square. %  with $\|w_k+\tilde{x}_{1|k}\|^2\geq 0$.  
 Combining~\eqref{proof_LMS_intermediate-a} and~\eqref{proof_LMS_intermediate-b} yields
    \begin{align}
        &\frac{1}{\mu} \|\hat{\theta}_{k+1}-\theta_{k+1}\|^2 - \frac{1}{\mu}\|\hat{\theta}_k-\theta_k\|^2\nonumber\\
        \stackrel{}{\leq} & -\|\tilde{x}_{1|k}\|^2 + \|w_k\|^2 + \frac{2c_\theta}{\mu}\|\Delta \theta_k\|.
    \end{align}
    Taking the sum from $k=0$ to $k=T$ on both sides, rearranging the terms, and using $0\leq\|\hat{\theta}_{T}-\theta_{T}\|^2$ yields~\eqref{probsetup_prop_LMS-bound-eq}.
    \qed

%%%%%%%%%%%%%%%%%
\subsubsection{Proof of Proposition~\ref{prop_constr-term-ingred}}
This proof consists of three parts: First, we construct a linear-quadratic bound on the stage cost. Then, we show that the terminal cost~\eqref{lin-quad_terminal-cost} satisfies~\eqref{assump_terminal-cost-decrease-eq} using~\eqref{numex_Lyapunov_inequality}. Lastly, we show that the terminal cost~\eqref{lin-quad_terminal-cost} is Lipschitz continuous, i.e., satisfies~\eqref{assump_Lf}. To improve the readability, we use the shorthand notation
    \begin{align*}
        q(x,v^\mathrm{s},\theta) &:= \left.\frac{\partial\bar{\ell}}{\partial x}\right|^\top_{(x,v^\mathrm{s},\theta)}.
    \end{align*}
Recall from~\eqref{shifted_stage_cost} that the shifted cost is defined as 
\begin{align*}% \label{shifted_stage_cost}
    \bar{\ell}(x, v^\mathrm{s}, \theta) := \ell(x,\kappa(x,h(v^\mathrm{s},\bar{\theta}),v^\mathrm{s})) - \ell^\mathrm{s}(v^\mathrm{s},\theta).
\end{align*}

Furthermore, recall from~\eqref{def_Zu_Zx-eq} that $\mathbb{Z}_x$ and $\mathbb{Z}_u$ are convex and satisfy $(x,u)\in\mathbb{Z} \implies x\in\mathbb{Z}_x,\ u\in\mathbb{Z}_u$.

\subsubsection*{Linear-quadratic stage cost bound}
First, we construct a linear-quadratic over-approximation $\ell_q$ of $\bar{\ell}$, similar to~\cite[Lemma~3]{kohlerPeriodicOptimalControl2020}. In particular, we show that
    \begin{align} 
    \begin{split}
        \ell_q (x,v^\mathrm{s},\theta) &:=
        q(x,v^\mathrm{s},\theta)\cdot (x-\hat{x}^s) + \frac{1}{2}\|x-\hat{x}^s\|^2_{Q} \\
        &\geq \bar{\ell}(x,v^\mathrm{s},\theta) + \frac{\alpha}{2} \|x-\hat{x}^s\|^2,\label{term-cost-ell-q}
    \end{split}
    \end{align}
    with $Q$ from~\eqref{Q_def-eq} and $\hat{x}^\mathrm{s}=h(v^\mathrm{s},\theta)$.
    Following the idea of~\cite[Lemma~23]{amritEconomicOptimizationUsing2011}, we use the mean value theorem for vector functions~\cite[Proposition A.11]{rawlingsModelPredictiveControl2020} to prove~\eqref{term-cost-ell-q}: Concretely, we apply it to the functional $\ell_q(x,v^\mathrm{s},\theta)-\bar{\ell}(x,v^\mathrm{s},\theta)$,  
    which is zero for $x=\hat{x}^{\mathrm{s}}$ by definition, 
    to relate the shifted cost at some point $x$ to that at the equilibrium $\hat{x}^s$. 
    This gives
    \begin{align*}
        &\ell_q(x,v^\mathrm{s},\theta) - \bar{\ell}(x,v^\mathrm{s},\theta) \nonumber \\
        &\quad = 
        \left[ q(x,v^\mathrm{s},\theta)-  \left.\frac{\partial \bar{\ell}}{\partial x}\right|_{(x,v^\mathrm{s},\theta)} \right]^\top\Delta x + \int_0^1 (1-s) \nonumber\\   
        &\qquad \cdot \Delta x^\top \left(Q-\left.\frac{\partial^2 \bar{\ell}}{\partial x^2}\right|_{(\hat{x}^s + s\Delta x,v^\mathrm{s},\theta)}\right) \Delta x\ ds \nonumber\\
        &\stackrel{\eqref{Q_def-eq},\eqref{term-cost-ell-q}}{\geq} \int_0^1 (1-s) \Delta x^\top \alpha I_{n_x} \Delta x\ ds \nonumber\\
        & \quad = \frac{\alpha}{2} \|x-\hat{x}^s\|^2, %\label{proof_constr-termingre_meanval}
    \end{align*}
    where $\Delta x:= x-\hat{x}^s,\ x\in\mathbb{Z}_x$.
    The inequality leveraged the fact that $\hat{x}^\mathrm{s} + s\Delta x\in\mathbb{Z}_x$ for $s\in[0,1]$, since $\mathbb{Z}_x$ is convex.   
    Consequently, the condition~\eqref{term-cost-ell-q} holds for all $x\in\mathbb{Z}_x$. 

\subsubsection*{Terminal cost decrease}
    Next, we show that the cost decrease~\eqref{assump_terminal-cost-decrease-eq} holds. 
    
    Consider the first-order Taylor approximation of the dynamics $f_\kappa(x,v^\mathrm{s},\theta)=f(x,\kappa(x,h(v^\mathrm{s},\bar{\theta}),v^\mathrm{s}),\theta)$ with the parameters $\theta$ around the steady-state $\hat{x}^s=h(v^\mathrm{s},\theta)$, which yields
    \begin{align}
    \begin{split}\label{proof_term-cost_Taylorexp-eq}
        f_\kappa(x,&\,v^\mathrm{s}, \hat{\theta}) - f_\kappa(\hat{x}^s,v^\mathrm{s}, \hat{\theta}) \\
        &= A_K(\hat{x}^\mathrm{s},v^\mathrm{s},\hat{\theta}) \cdot (x-\hat{x}^s) + r(x,\hat{x}^\mathrm{s},v^\mathrm{s},\hat{\theta}), 
        \end{split}
    \end{align}
    where $r(x,\hat{x}^\mathrm{s},v^\mathrm{s},\hat{\theta})$ is the remainder term and $A_K$ defined in~\eqref{offline_linearization_AK-def-eq}. By using the Lagrange form of the approximation error, we can upper bound its components quadratically
    \begin{align}
    \label{remainder_H-bound}
        \left|[r(x,\hat{x}^s,v^\mathrm{s},\theta)]_i \right| \leq \frac{H}{2} \|x-\hat{x}^s\|^2,
    \end{align}
    where $H$ according to~\eqref{term-cost_H_def-eq} is an upper bound on the maximum eigenvalue of the Hessian of the closed-loop dynamics $\frac{\partial^2 [f_\kappa]_i}{\partial x^2}$. 

    In the following, we abbreviate $x^+:=f_\kappa(x,v^\mathrm{s},\theta)$, $\tilde{A}_K:=A_K(\hat{x}^s,v^\mathrm{s},\theta)$, and $\tilde{q}:=q(\hat{x}^\mathrm{s},v^\mathrm{s},\theta)$. For the linear terms of the cost function, we  obtain
    {\allowdisplaybreaks
    \begin{align}
        & p(\theta, \hat{x}^s, v^\mathrm{s})^\top \left[ (x^+-\hat{x}^\mathrm{s}) - (x-\hat{x}^\mathrm{s}) \right]\nonumber\\
        %%%%
        & \!\! = p(\theta, \hat{x}^s, v^\mathrm{s})^\top \left[ (x^+ -\hat{x}^s) - (\tilde{A}_K\cdot (x - \hat{x}^s))
          \right.\nonumber\\
        &\quad \left. + (\tilde{A}_K\cdot (x - \hat{x}^s)) - (x-\hat{x}^\mathrm{s})\right]
        \nonumber \\
        %%%%
        &\bquad\, \stackrel{\eqref{proof_term-cost_Taylorexp-eq}}{=} p(\theta, \hat{x}^s, v^\mathrm{s})^\top \left(\!
        r(x, \hat{x}^\mathrm{s}, v^\mathrm{s}, \theta) + (\tilde{A}_K-I_{n_x}) \cdot (x-\hat{x}^\mathrm{s})\!\right) \nonumber \\
        %%%%
        &\!\!\! \stackrel{\eqref{numex_p-vector-def}}{=}
        p(\theta, \hat{x}^s, v^\mathrm{s})^\top 
        r(x, \hat{x}^\mathrm{s}, v^\mathrm{s}, \theta) - \tilde{q}^\top(x-\hat{x}^\mathrm{s}) \nonumber\\
        &\!\!\!\stackrel{\eqref{remainder_H-bound}}{\leq} \frac{\sqrt{n_x}H}{2} \|p(\theta, \hat{x}^s, v^\mathrm{s})\|\cdot \|x-\hat{x}^\mathrm{s}\|^2 - \tilde{q}^\top (x-\hat{x}^\mathrm{s}), \label{numex_prop_proof-b}
     \end{align}
     }
     where the last inequality also used $\|r\| \leq \sqrt{n_x}\|r\|_\infty$.
Furthermore, the quadratic terms of the terminal cost satisfies
\begin{align}
    \|x^+-\hat{x}^\mathrm{s}\|^2_{P_\mathrm{f}} &- \|x-\hat{x}^\mathrm{s}\|^2_{P_\mathrm{f}} \nonumber\\
        &= \|f_\kappa(x,v^\mathrm{s},\theta)-\hat{x}^\mathrm{s}\|^2_{P_\mathrm{f}} - \|x-\hat{x}^\mathrm{s}\|^2_{P_\mathrm{f}} \nonumber\\
    &\!\!\stackrel{\eqref{numex_Lyapunov_inequality}}{\leq} -\|x-\hat{x}^\mathrm{s}\|^2_{Q}, \label{numex_prop_proof-c}
\end{align}
    where the inequality holds by using contraction theory~\cite[Theorem~2.8]{tsukamotoContractionTheoryNonlinear2021} and that $\mathbb{Z}_x$ is convex. 
    Combining the results from~\eqref{numex_prop_proof-b} and~\eqref{numex_prop_proof-c}, yields
\begin{align*}
    &\ell_\mathrm{f}(x^+,\hat{x}^\mathrm{s},v^\mathrm{s},\theta) - \ell_\mathrm{f}(x,\hat{x}^\mathrm{s},v^\mathrm{s},\theta)\nonumber\\
    %%%%%%
    &\ \stackrel{\eqref{numex_prop_proof-b},\eqref{numex_prop_proof-c}}{\leq } -\|x-\hat{x}^\mathrm{s}\|^2_{Q} + \frac{\sqrt{n_x}H}{2} \|p(\theta, \hat{x}^s, v^\mathrm{s})\| \cdot \|x-\hat{x}^\mathrm{s}\|^2 \nonumber\\
    & \qquad - \tilde{q}^\top (x-\hat{x}^\mathrm{s})\nonumber\\
    %%%%%%
    &\,\stackrel{\eqref{term-cost-ell-q},\eqref{alpha_def-eq}}{\leq} -\bar{\ell}(x,v^\mathrm{s},\theta) \\
    &\quad\stackrel{\eqref{shifted_stage_cost}}{=} -\ell(x,\kappa(x,h(v^\mathrm{s},\bar{\theta}),v^\mathrm{s})) + \ell^\mathrm{s}(v^\mathrm{s},\theta).
\end{align*}
Hence, the cost decrease property~\eqref{assump_terminal-cost-decrease-eq} holds.

\subsubsection*{Lipschitz-continuity of $\ell_\mathrm{f}$}
The vector $p$ is uniformly bounded since $\tilde{A}_K$ is Schur stable with a uniform bound and the gradient of $\bar{\ell}$ is bounded thanks to the Lipschitz continuity of $f_\mathrm{w}$, $\kappa$, and $\ell$. Together with the compact parameter set $\Theta$ and the compact constraints $\mathbb{Z}$, this implies that the terminal cost~\eqref{lin-quad_terminal-cost} is Lipschitz-continuous, i.e.,~\eqref{assump_Lf} holds. 

Consequently, Assumption~\ref{assump_terminal-ingredients} holds.
\qed

%%%%%%%%%%%%%
\subsubsection{Proof of Theorem~\ref{thm_transient}} 
This proof consists of three parts: First, we prove that Problem~\eqref{MPC_scheme} is recursively feasible. Second, we relate the change of the cost function to the change of parameter estimates. Lastly, we make use of the LMS parameter estimator's properties (Lemma~\ref{lemma_accum_x_tilde}) to obtain our performance bound.     
    \paragraph*{Recursive feasibility} 
    We show recursive feasibility by constructing a feasible candidate solution to Problem~\eqref{MPC_scheme} at time $k+1$, given a feasible solution at time $k$. 
    Let us denote
    \begin{subequations}
        \begin{align*}
           z_{N+1|k}^\ast &= z_{N|k}^\ast, \qquad\quad\ \ v_{N|k}^\ast :=  v^{\mathrm{s}\ast}_k,\\%\label{transient_candidate-c}\\
        \mathbb{X}_{N+1|k}^\ast &:= \mathbb{X}_{N|k}^\ast,\qquad\ \, \hat{x}_{0|k+1} = x_{k+1}\\ %\label{transient_candidate-d}
        \hat{x}_{N+1|k}^\ast &= f(\hat{x}_{N|k}^\ast, \kappa(\hat{x}_{N|k}^\ast, z_{N|k}^\ast, v^\ast_{N|k}),\hat{\theta}_k).
        \end{align*}
    \end{subequations}
    The considered candidate solution is given by
    \begin{subequations}\label{transient_candidate}\begin{alignat}{2} 
        v_{j|k+1} &= v^\ast_{j+1|k},   && \ j \in\mathbb{I}_{[0,N-1]}\label{transient_candidate-a}\\
        v^\mathrm{s}_{k+1} &= v^{\mathrm{s}\ast}_k,\quad\ \, v_{N|k+1} = v^s_{k+1}, && \label{transient_candidate-g}\\
        \mathbb{X}_{j|k+1} &= \mathbb{X}^\ast_{j+1|k},\  z_{j|k+1} = z_{j+1|k}^\ast,  && \ j \in\mathbb{I}_{[1,N]}\label{transient_candidate-b}\\
        \hat{x}_{j+1|k+1} &= f(\hat{x}_{j|k+1}, \hat{u}_{j|k+1}, \hat{\theta}_{k+1}),      &&\  j \in\mathbb{I}_{[0,N]} \label{transient_candidate-e}\\
        \hat{u}_{j|k+1} &= \kappa(\hat{x}_{j|k+1},z_{j|k+1},v_{j|k+1}),\ \                 &&\ j \in\mathbb{I}_{[0,N]}.\label{transient_candidate-f}
    \end{alignat}\end{subequations}
    We shift the nominal trajectory  $(z^\ast_{\cdot|k},v^\ast_{\cdot|k})$ and the prediction tubes $\mathbb{X}_{\cdot|k}^\ast$ in~\eqref{transient_candidate-a}--\eqref{transient_candidate-b} and choose $\mathbb{X}_{N|k+1} = \mathbb{X}^\ast_{N|k}$, which guarantees that the nominal candidate trajectory satisfies the constraints~\eqref{robust-MPC_scheme-b}, \eqref{robust-MPC_scheme-c}, and~\eqref{robust-MPC_scheme-d} at time $k+1$. The satisfaction of~\eqref{MPC_scheme-g} is guaranteed by shifting the tube with~\eqref{transient_candidate-b} and thanks to the tube propagation property~\eqref{assump_Phi-function-eq}, i.e., $x_{k+1},z_{1|k}^\ast\in\mathbb{X}^\ast_{1|k}$ due to $\theta_k,\bar{\theta}\in\Theta$.

    The constraints~\eqref{MPC_scheme-d} and~\eqref{MPC_scheme-e} are satisfied by definition~\eqref{transient_candidate-e}--\eqref{transient_candidate-f}. As the last nominal state is $z_{N|k+1}=z^\ast_{N|k}$ and the steady-input is $v^\mathrm{s}_{k+1}=v^\mathrm{s\ast}_k$, constraint~\eqref{robust-MPC_scheme-f2} holds. 
    The constraint~\eqref{robust-MPC_scheme-f} ensures that the last predicted tube is robust positive invariant and choosing $\mathbb{X}_{N|k+1} = \mathbb{X}^\ast_{N|k}$ guarantees satisfaction of~\eqref{robust-MPC_scheme-f} and~\eqref{robust-MPC_scheme-f2} at the next time step. 
    Inequality~\eqref{MPC_scheme-j} holds with equality for the candidate solution due to~\eqref{eq:lambda_k}.

    In the following, we abbreviate
    \begin{alignat*}{1}
        \ell_{j|k+1} &:= \ell(\hat{x}_{j|k+1}, \hat{u}_{j|k+1}), \quad  \ell^\ast_{j|k} \,:=\, \ell(\hat{x}^\ast_{j|k},\hat{u}^\ast_{j|k}), \\
        \Delta\hat{\theta}_k &:= \hat{\theta}_{k+1} - \hat{\theta}_k,  
    \end{alignat*}
for $j \in\mathbb{I}_{[0,N]}$, and note that $\ell_{0|k}^\ast = \ell(x_k,u_k)$. Recall from Section~\ref{ssec:LMS} that $\tilde{x}_{1|k} +w_k= x_{k+1}-\hat{x}_{1|k}^\ast$ and $\Delta\theta_k = \theta_{k+1} - \theta_k$.
    \paragraph*{Bounding the change of the cost function}
Based on the feasible candidate~\eqref{transient_candidate}, the optimal cost $\mathcal{J}^\ast_N$ satisfies
\begin{align} \label{transient-proof_eq-a}
\begin{split}
    \mathcal{J}^\ast_N(&x_{k+1},\hat{\theta}_{k+1}, \lambda_{k+1})-\mathcal{J}^\ast_N(x_k,\hat{\theta}_k, \lambda_k) +\ell^\ast_{0|k} \\
    %%%%%%%%%%%%%%%%%%%%%%%%%%%%%%
     &\leq \sum_{j=1}^{N-1} \left(\ell_{j-1|k+1} - \ell_{j|k}^\ast\right) - \ell_\mathrm{f}(\hat{x}^\ast_{N|k}, \hat{x}^\mathrm{s\ast}_k, v^\mathrm{s\ast}_k, \hat{\theta}_k) \\
    &\quad +\ell_{N-1|k+1} + \ell_\mathrm{f}(\hat{x}_{N|k+1}, \hat{x}^\mathrm{s}_{k+1}, v^\mathrm{s}_{k+1}, \hat{\theta}_{k+1})  \\
    & \quad + \beta \ell^\mathrm{s}(v^\mathrm{s}_{k+1},\hat{\theta}_{k+1}) -\beta \ell^\mathrm{s}(v^{\mathrm{s}\ast}_k,\hat{\theta}_k). 
\end{split}
\end{align}
By using the terminal cost's property~\eqref{assump_terminal-cost-decrease-eq} and the Lipschitz continuity of $\ell_\mathrm{f}$ and $h$, we can bound the terminal cost terms as follows 
{\allowdisplaybreaks
\begin{align}
    & \ell_\mathrm{f}(\hat{x}_{N|k+1}, \hat{x}^\mathrm{s}_{k+1}, v^\mathrm{s}_{k+1}, \hat{\theta}_{k+1})-\ell_\mathrm{f}(\hat{x}^\ast_{N|k}, \hat{x}^\mathrm{s\ast}_k, v^\mathrm{s\ast}_k, \hat{\theta}_k)   \nonumber \\
    & \ \ \stackrel{\eqref{assump_terminal-cost-decrease-eq}}{\leq}  -\ell^\ast_{N|k} -\ell_\mathrm{f}(\hat{x}^\ast_{N+1|k},\hat{x}^\mathrm{s\ast}_k, v^\mathrm{s\ast}_k, \hat{\theta}_k)\nonumber\\
        &\qquad + \ell_\mathrm{f}(\hat{x}_{N|k+1}, \hat{x}^\mathrm{s}_{k+1}, v^\mathrm{s}_{k+1}, \hat{\theta}_{k+1}) + \ell^\mathrm{s}(v^{\mathrm{s}\ast}_k,\hat{\theta}_k) \nonumber\\
    & \ \ \stackrel{\eqref{assump_Lf}}{\leq} -\ell^\ast_{N|k}  +  L_\mathrm{f} \left(  \|\hat{x}_{N|k+1} - \hat{x}^\ast_{N+1|k}\| + \|v^{\mathrm{s}}_{k+1}-v^\mathrm{s\ast}_k\| \right. \nonumber \\
        & \qquad \left.  +\| h(v^{\mathrm{s}}_{k+1}, \hat{\theta}_{k+1}) - h(v^{\mathrm{s}\ast}_k, \hat{\theta}_k) \| + \|\Delta\hat{\theta}_k\|\right) + \ell^\mathrm{s}(v^{\mathrm{s}\ast}_k,\hat{\theta}_k) \nonumber\\
    & \stackrel{\eqref{assump_equilib-manifold-eq},\eqref{transient_candidate-g}}{\leq} -\ell(\hat{x}^\ast_{N|k},\kappa(\hat{x}^\ast_{N|k}, z^\ast_{N|k}, v^{\mathrm{s}\ast}_k)) + \ell^\mathrm{s}(v^{\mathrm{s}\ast}_k,\hat{\theta}_k) \nonumber\\  
    &\qquad +L_\mathrm{f} (1+ L_\mathrm{h}) \|\Delta\hat{\theta}_k\|  +  L_\mathrm{f} \|\hat{x}_{N|k+1} - \hat{x}^\ast_{N+1|k}\|.
    \label{transient-proof_eq-b}
\end{align}
}
\paragraph*{Combination with LMS properties}
Recall that $\ell^\mathrm{s}(v,\theta) = \ell(h(v,\theta),\ \kappa_\mathrm{s}(v,\theta))$.
For any $v\in\mathbb{Z}_u$ and any $\theta,\theta'\in\Theta$, it holds that
\begin{align}
        \ell^\mathrm{s}(v,\theta) - \ell^\mathrm{s}(v,\theta') 
        &\stackrel{\eqref{assump_regularity-eq-b}}{\leq} L_{\ell} \|h(v,\theta) - h(v,\theta')\| \nonumber\\
        &\stackrel{\eqref{assump_equilib-manifold-eq}}{\leq} L_\mathrm{h} L_{\ell} \|\theta-\theta'\|=: L_\mathrm{s} \|\theta-\theta'\|,\label{transient-proof_eq-c}
    \end{align}
    where $L_\mathrm{s}\geq0$. 
The inequality 
 \begin{equation} \label{transient-proof_eq-d}
            \|\hat{x}_{j-1|k+1} - \hat{x}^\ast_{j|k}\|   \leq   C_{\mathrm{dyn},\kappa} \cdot \|\tilde{x}_{1|k} + w_k\|~~\forall j\in\mathbb{I}_{[1,N+1]},
        \end{equation}
holds with a uniform constant $C_{\mathrm{dyn},\kappa}\geq 0$ thanks to the Lipschitz continuity of the dynamics~\eqref{assump_regularity-eq-a} and the feedback $\kappa$~\eqref{assump_regularity-eq-c}, see Proposition~\ref{prop_consecutive-distances} in Appendix~\ref{appendix_aux-results}. 
Combining the intermediate results~\eqref{transient-proof_eq-a}--\eqref{transient-proof_eq-d} yields
\begin{align}
     \mathcal{J}_N^\ast(&x_{k+1},\hat{\theta}_{k+1}, \lambda_{k+1})-\mathcal{J}^\ast_N(x_k,\hat{\theta}_k, \lambda_k) +\ell^\ast_{0|k} \nonumber\\
     %%%%%%%%%%%%%%%%%%%%%%%%%%%%%%
     & \bquad \bquad \stackrel{\eqref{transient-proof_eq-a}-\eqref{transient-proof_eq-c}}{\leq}
     \sum_{j=1}^{N}\left(\ell_{j-1|k+1} - \ell_{j|k}^\ast\right) 
     + L_\mathrm{f} (1+ L_\mathrm{h}) \|\Delta\hat{\theta}_k\| \nonumber\\
     &  \qquad +  L_\mathrm{f} \|\hat{x}_{N|k+1} - \hat{x}^\ast_{N+1|k}\| 
     + (1- \beta) \ell^\mathrm{s}(v^{\mathrm{s}\ast}_k,\hat{\theta}_k) \nonumber\\
     &\qquad + \beta \cdot \left( L_\mathrm{s} \| \Delta\hat{\theta}_k\| + \ell^\mathrm{s}(v^{\mathrm{s}\ast}_k,\hat{\theta}_k)\right) \nonumber \\ 
     %%%%%%%%%%%%%%%%%%%%%%%%%%%%%%
     &\bquad\bquad\stackrel{\eqref{transient-proof_eq-d},\eqref{assump_regularity-eq-b}}{\leq} 
     C_{\mathrm{dyn},\kappa}(L_\ell(N-1)+L_\mathrm{f})\|\tilde{x}_{1|k}+w_k\| 
     \nonumber\\
     &\qquad  + (L_\mathrm{f} (1+ L_\mathrm{h})+\beta L_\mathrm{s}) \|\Delta\hat{\theta}_k\| +  \ell^\mathrm{s}(v^{\mathrm{s}\ast}_k,\hat{\theta}_k)  
     \nonumber \\
     %%%%%%%%%%%%%%%%%%%%%%%%%%%%%%
    &\bquad \!\stackrel{\eqref{probsetup_LMS_gain-definition},\eqref{probsetup_LMS-bound_proof_aux-result}}{\leq} (C_A + \sqrt{\mu}\beta L_\mathrm{s}) \|\tilde{x}_{1|k}+w_k\|  +  \ell^\mathrm{s}(v^{\mathrm{s}\ast}_k,\hat{\theta}_k),
    \label{thm_transient-proof_intermediate}
\end{align}
with
\begin{align*}
    C_A &:= C_{\mathrm{dyn},\kappa}(L_\ell(N-1)+L_\mathrm{f}) + \sqrt{\mu}L_\mathrm{f} (1+ L_\mathrm{h}). 
\end{align*}
    Thus, taking a telescopic sum of~\eqref{thm_transient-proof_intermediate} and using the bounds on the cumulative prediction error and disturbances from Lemma~\ref{lemma_accum_x_tilde} yields
    \begin{align*}
        &\mathcal{J}^\ast_N(x_T,\hat{\theta}_T, \lambda_T) - \mathcal{J}^\ast_N(x_0,\hat{\theta}_0, \lambda_0) + \sum_{k=0}^{T-1} \ell(x_k,u_{k})\\
        & \ \stackrel{\eqref{thm_transient-proof_intermediate}}{\leq}
        \sum_{k=0}^{T-1}\left[\ell^\mathrm{s}(v^\mathrm{s}_k,\hat{\theta}_k) + (C_A + \sqrt{\mu}\beta L_\mathrm{s}) \|\tilde{x}_{1|k}+w_k\|\right] \nonumber\\
        %%%%
        & \!\!\stackrel{\eqref{proof_w-bound-eq},\eqref{lemma_accum_x_tilde-eq}}{\leq}
        \! \sum_{k=0}^{T-1}\left[\ell^\mathrm{s}(v^\mathrm{s}_k,\hat{\theta}_k)\right] + \sqrt{T}(C_A + \sqrt{\mu}\beta L_\mathrm{s})\!\cdot\! \left[2|w_{[0,T-1]}|_{\mathcal{L}_2}\vphantom{\sqrt{\frac{1}{\mu}}}\right. \nonumber\\
        & \qquad +\left.\sqrt{\frac{1}{\mu}}\|\hat{\theta}_0-\theta_0\| + \sqrt{\frac{2c_\theta}{\mu}}\sum_{k=0}^{T-1}  \sqrt{ \|\Delta\theta_k\|}\right].\nonumber
    \end{align*}
    A uniform constant 
    \begin{align} \label{eq:C_J-definition}
        -C_\mathcal{J} \leq \mathcal{J}^\ast_N(x_T,\hat{\theta}_T, \lambda_T) - \mathcal{J}^\ast_N(x_0,\hat{\theta}_0, \lambda_0)
    \end{align}
    \!exists due to the boundedness of $\ell(x,u)$ and $\ell_\mathrm{f}(x,z^\mathrm{s}_k,\theta)$ on the compact set~$\mathbb{Z}$ and due to $\beta < \infty$.
    Reordering the terms in the last inequality and using the constant $C_{\mathcal{J}}$ results in~\eqref{thm_transient-eq}.
    \qed

%%%%%%%%%%%%
\subsubsection{Proof of Corollary~\ref{corr_asymptotic_performance}}
    Note that $\lambda_\infty=\limsup_{k\rightarrow\infty}\lambda_k$ is finite because $\lambda_k$ is bounded due to compact $\mathbb{Z},\Theta$. Similar to~\cite[Theorem~1]{mullerEconomicModelPredictive2013}, we define the auxiliary variable $\zeta_k:= -\lambda_\infty +  \ell^\mathrm{s}(v^{\mathrm{s}\ast}_k,\hat{\theta}_k)$.  
    Using  Conditions~\eqref{MPC_scheme-j} and \eqref{eq:lambda_k} and Lipschitz continuity~\eqref{transient-proof_eq-c} gives 
    \begin{align}\label{lambda_decrease_eq}
    \lambda_{k+1}\leq \lambda_k+L_{\mathrm{s}}\|\Delta\hat{\theta}_k\|.
    \end{align}
    Hence, 
    \begin{align}
        %\liminf_{T\rightarrow\infty} \zeta_T \leq 
        &\limsup_{T\rightarrow\infty} \dfrac{1}{T}\sum_{k=0}^{T}\zeta_k \nonumber\\
        &\bquad\!\stackrel{\eqref{eq:lambda_k},\eqref{transient-proof_eq-c}}{\leq} \limsup_{T\rightarrow\infty} \dfrac{1}{T}\sum_{k=0}^{T}\lambda_{k+1}-\lambda_\infty + L_\mathrm{s}\|\Delta\hat{\theta}_k\|\label{Stolz-Cesaro_bound}\\
        &\leq \limsup_{k\rightarrow\infty} \lambda_{k+1}-\lambda_\infty + L_\mathrm{s}\|\Delta\hat{\theta}_k\|=\lim_{k\rightarrow\infty}L_\mathrm{s}\|\Delta\hat{\theta}_k\|=0, \nonumber
    \end{align}
    where the second inequality holds using the Stolz-Cesàro-Theorem~\cite[Section~3.4]{tothRationalRealExponentiation2021} and $\|\Delta \hat{\theta}_k\|\rightarrow0$ for $k\rightarrow \infty$ follows from~\eqref{probsetup_prop_LMS-bound-eq}--\eqref{probsetup_prop_LMS-theta-diff-eq} and Assumption~\ref{assump_finite-energy}.   
    
    Thus, the intermediate result~\eqref{thm_transient-proof_intermediate} yields
    \begin{align}\label{asympt_proof_final}
        &\limsup_{T\rightarrow\infty} \left\{ \frac{\mathcal{J}^\ast_N(x_{T}, \hat{\theta}_{T}, \lambda_T)\! -\! \mathcal{J}_N^\ast(x_{0}, \hat{\theta}_{0}, \lambda_0)}{T}+\frac{\!\sum_{k=0}^{T-1} \ell_{0|k}^\ast}{T}\right\}\nonumber \\
        \!&\stackrel{\eqref{thm_transient-proof_intermediate}}{\leq} \limsup_{T\rightarrow\infty} \frac{1}{T}\sum_{k=0}^{T-1} \left[   (C_A\!+\!\sqrt{\mu}\beta L_s) \|\Tilde{x}_{1|k}+w_k\| \!+\!\ell^s(v^{\mathrm{s}\ast}_k,\hat{\theta}_k) \right]\nonumber \\ 
        &\stackrel{\eqref{lemma_asymp-lin-bound_equation}}{=} \limsup_{T\rightarrow\infty} \frac{1}{T}\left[ T\lambda_\infty + \sum_{k=0}^{T-1} \zeta_k \right]  
     \stackrel{\eqref{Stolz-Cesaro_bound}}{\leq}\lambda_\infty. 
    \end{align} 
The cost function~\eqref{MPC_scheme-a} is uniformly bounded on $\mathbb{Z}$. Thus, $\mathcal{J}^\ast_N(x_k,\hat{\theta}_k, \lambda_k)$ is also bounded and 
\begin{equation} \label{proof_limitJ}
    \limsup_{T\rightarrow\infty} \frac{\mathcal{J}^\ast_N(x_0,\hat{\theta}_0, \lambda_0) - \mathcal{J}^\ast_N(x_T,\hat{\theta}_T, \lambda_T) }{T} =0.    
\end{equation}
By reordering the terms in~\eqref{asympt_proof_final} and using~\eqref{proof_limitJ}, we obtain the desired performance bound~\eqref{corr_asymptotic_performance-eq}.
\qed

\subsubsection{Proof of Corollary~\ref{corr_subopt}}
    The first inequality in~\eqref{corr_subopt-eq-b} is equivalent to~\eqref{corr_asymptotic_performance-eq} in Corollary~\ref{corr_asymptotic_performance}. 
    In the following, we prove the second inequality by adapting the proof  from~\cite{fagianoGeneralizedTerminalState2013}. 
    
    Denote the optimal solution to Problem~\eqref{MPC_scheme} at time $k$ by ${v}^\ast_{\cdot|k}$ and $\ell^\ast:=\ell(\hat{x}^\mathrm{s\ast}_k,\hat{u}^\mathrm{s\ast}_k)$, which exists due to recursive feasibility. By definition~\eqref{def_opt_achiev_cost}, there also exists a feasible sequence $\tilde{v}_{\cdot|k}$ with a setpoint $\ell(\hat{\tilde{x}}^\mathrm{s}_k, \hat{\tilde{u}}^{\mathrm{s}}_k) = \underline{\ell}(x_k,\hat{\theta}_k, \lambda_k)=:\underline{\ell}$. 
     Assume for contradiction that 
     \begin{align}\label{subopt_proof_contradict-eq}
         \ell^\ast > \underline{\ell}+\varepsilon.
     \end{align} 
     Let us denote the cost of the feasible candidate solution and the optimal solution by $\tilde{\mathcal{J}}+\beta\underline{\ell}$ and $\mathcal{J}^\ast+\beta \ell^\ast$, respectively.  
     Given $\mathbb{Z}$ compact and Lipschitz-continuity of the dynamics, cost, and feedback (Assumptions~\ref{assump_regularity} and~\ref{assump_terminal-ingredients}), the difference of the feasible open-loop costs is bounded by some uniform constant $\eta>0$, i.e., $|\mathcal{\tilde{J}}-\mathcal{J}^\ast|\leq \eta$. 
     Thus, we arrive at the following contradiction
     \begin{align*}
       \mathcal{J}^\ast + \beta \ell^\ast\leq &\tilde{\mathcal{J}}+\beta \underline{\ell}\leq \eta+\mathcal{J}^\ast + \beta \underline{\ell}
     \stackrel{\eqref{subopt_proof_contradict-eq}}{<} \eta+\mathcal{J}^\ast + \beta (\ell^\ast - \varepsilon)\\
      \leq&  \mathcal{J}^\ast + \beta \ell^\ast.
     \end{align*}
    The first inequality holds because $v^\ast_{\cdot|k}$ is the optimal solution of the MPC problem and thus cannot have a larger cost than the feasible candidate $\tilde{v}_{\cdot|k}$. 
    The last inequality holds by setting $\beta\geq \underline{\beta}(\epsilon):=\eta/\epsilon$. 
     Given this contradiction,~\eqref{subopt_proof_contradict-eq} cannot hold. Thus
    $\ell^\ast\leq \underline{\ell}+\varepsilon$, i.e., Inequality~\eqref{corr_subopt-eq-b} holds.    
\qed

%%%%%%%%%%%%%%%%%%%%%%%%%%%%%%%%%%%%%%%%%%%
\subsection{Auxiliary Lemmas} \label{appendix_aux-results}

\subsubsection*{Proof of Lemma~\ref{lemma_asymp-lin-bound}}
    We use Lemma~\ref{lemma_infinite-series} from Appendix~\ref{appendix_aux-results} with $a_k = \|w_k\|^2$ and $\alpha(s) = \sqrt{s}, \alpha\in \mathcal{K}_\infty$ to get
    \begin{align}
        \lim_{T\rightarrow\infty}\sum_{k=0}^{T-1} \frac{\sqrt{\|w_k\|^2}}{T} = \lim_{T\rightarrow\infty}\sum_{k=0}^{T-1} \frac{\|w_k\|}{T}= 0  \label{lemma_bound_proof_a}.
    \end{align}
    Furthermore, $\Theta$ compact implies $\frac{1}{\mu}\|\hat{\theta}_0-\theta_0\|^2<\infty$. Together with Assumption~\ref{assump_finite-energy} and Proposition~\ref{probsetup_prop_LMS-bound} this implies
    \begin{align*}
       \sum_{k=0}^{T-1} \|\tilde{x}_{1|k}\|^2 \stackrel{\eqref{probsetup_prop_LMS-bound-eq}}{\leq} S_{\mathrm{w}} + \frac{1}{\mu}\|\hat{\theta}_0-\theta_0\|^2 + \frac{2c_\theta}{\mu} S_\theta < \infty.
        \end{align*}
    Applying  Lemma~\ref{lemma_infinite-series} 
    with $a_k = \|\tilde{x}_{1|k}\|^2$ yields   
    \begin{align}
         \lim_{T\rightarrow\infty}\sum_{k=0}^{T-1} \frac{\sqrt{\|\tilde{x}_{1|k}\|^2}}{T}
        =\! \lim_{T\rightarrow\infty}\sum_{k=0}^{T-1} \frac{\|\tilde{x}_{1|k}\|}{T} = 0  .\label{lemma_bound_proof_b}
    \end{align}
    Adding~\eqref{lemma_bound_proof_a} and~\eqref{lemma_bound_proof_b} yields the desired inequality. 
    \qed

%%%%%%%%%%%%%%%%%%

\begin{lemma}[Accumulated disturbance and prediction error]\label{lemma_accum_x_tilde}
   Let Assumption~\ref{assump_bounded-disturb-param} hold. Suppose $(x_k,u_k)\in\mathbb{Z}$ for all $k\in\mathbb{N}$. Then, for all $T\in \mathbb{N}$, it holds that
    \begin{align}
    \begin{split}
        \sum_{k=0}^{T-1} \|w_k\|
            &\leq \sqrt{T}
            |w_{[0,T-1]}|_{\mathcal{L}_2}, \label{proof_w-bound-eq}
    \end{split}\\
    \begin{split}
        \sum_{k=0}^{T-1} \|\tilde{x}_{1|k}\| 
        & \leq \sqrt{T}|w_{[0,T-1]}|_{\mathcal{L}_2} + \sqrt{\frac{T}{\mu}} \|\hat{\theta}_0-\theta_0\| \\
        &\qquad + \sqrt{\frac{2T c_\theta}{\mu}} \sum_{k=0}^{T-1}  \sqrt{ \|\Delta\theta_k\|}. \label{lemma_accum_x_tilde-eq}
        \end{split}
    \end{align}
\end{lemma}
\begin{proof}
Using the equivalence of norms, we have
    \begin{align*}
            \sum_{k=0}^{T-1}\|w_k\|
            &\leq \sqrt{T}\sqrt{\sum_{k=0}^{T-1} \|w_k\|^2}\nonumber\\
            &=
            \sqrt{T}\sqrt{|w_{[0,T-1]}|_{\mathcal{L}_2}^2} = \sqrt{T}|w_{[0,T-1]}|_{\mathcal{L}_2},
    \end{align*}  
    where we used that $\|W\|_1\leq \sqrt{T}\|W\|$ with $[W]_i = \|w_i\|$.
    Similarly, we obtain 
    {\allowdisplaybreaks
    \begin{align*}
        &\sum_{k=0}^{T-1} \|\tilde{x}_{1|k}\| 
        \leq \sqrt{T}\sqrt{\sum_{k=0}^{T-1}\|\tilde{x}_{1|k}\|^2} \\
        \stackrel{\eqref{probsetup_prop_LMS-bound-eq}}{\leq}& \sqrt{T}\sqrt{\frac{1}{\mu}  \|\hat{\theta}_0 - \theta_0\|^2 
        +\sum_{k=0}^{T-1} \left[\|w_k\|^2 + \frac{2c_\theta}{\mu} \|\Delta\theta_k\| \right]} \\
         \stackrel{\eqref{proof_w-bound-eq}}{\leq}& \sqrt{T}|w_{[0,T-1]}|_{\mathcal{L}_2}+\sqrt{\frac{T}{\mu}} \left(\|\hat{\theta}_0-\theta_0\| \!+\!  \sqrt{2c_\theta\sum_{k=0}^{T-1}   \|\Delta\theta_k\|}\right) \\
         \leq &\sqrt{T}|w_{[0,T-1]}|_{\mathcal{L}_2} \!+\! \sqrt{\frac{T}{\mu}}\left(\! \|\hat{\theta}_0-\theta_0\| \!+\! \sqrt{2c_\theta}\sum_{k=0}^{T-1}  \sqrt{ \|\Delta\theta_k\|} \!\right)\!, 
    \end{align*}    
    }
where we used ${\sqrt{a+b} \leq \sqrt{a}+\sqrt{b}}$ in the last inequality.
\end{proof} 

\begin{proposition}[Distance of consecutive state predictions]
\label{prop_consecutive-distances} 
    Suppose the conditions in Theorem~\ref{thm_transient} hold. Then, for any time $k\in\mathbb{I}_{\geq 0}$, the candidate~\eqref{transient_candidate} solution satisfies for all $j\in\mathbb{I}_{[1,N+1]}$
        \begin{equation} \label{prop_prelim-result-ineq}
            \|\hat{x}_{j-1|k+1} - \hat{x}_{j|k}\| \leq 
            C_{\mathrm{dyn},\kappa} \cdot \|\tilde{x}_{1|k} + w_k\|,
        \end{equation}
    with $C_{\mathrm{dyn},\kappa}:= \sum_{i=0}^{N} \left[L_\mathrm{dyn}(1+L_\kappa)\right]^i>0$.

\end{proposition}
\begin{proof}
    We use induction to show 
    \begin{align*}
        \|\hat{x}_{j-1|k+1} - \hat{x}^\ast_{j|k}\| \leq C(j)\cdot  \|\tilde{x}_{1|k} + w_k\|,~j\in\mathbb{I}_{[1,N+1]}
    \end{align*}
 with $C(j):= \sum_{i=0}^{j-1} \left[L_\mathrm{dyn}(1+L_\kappa)\right]^i$. 
    
    \emph{Base case ($j=1$): } The base case holds by definition with equality for $C(1)=1$ since
    \begin{equation*}
        \|\hat{x}_{0|k+1}-\hat{x}^\ast_{1|k}\| = \|x_{k+1}-\hat{x}^\ast_{1|k}\| \stackrel{\eqref{prop_prelim-result-ineq}}{=} 1\cdot \|\tilde{x}_{1|k}+w_k\|.
    \end{equation*}
    
    \emph{Induction step ($j=\lambda \rightarrow j=\lambda+1$): } Suppose that for some $j=\lambda \in \mathbb{I}_{\geq 1}$ the inequality
    \begin{equation} \label{prop_prelim_induction-step}
        \|\hat{x}_{\lambda-1|k+1}-\hat{x}^\ast_{\lambda|k}\| \leq C(\lambda)
\|\tilde{x}_{1|k}+w_k\|
    \end{equation}
    holds. Recall from~\eqref{probsetup_prop_LMS-theta-diff-eq} that 
    \begin{align*}
     \|\Delta\hat{\theta}_k \|:=  \|\hat{\theta}_{k+1} - \hat{\theta}_k\|  
        \leq \sqrt{\mu} \cdot \|\tilde{x}_{1|k}+w_k\|. 
    \end{align*}
    Then, for $j=\lambda+1$, inequality~\eqref{prop_prelim_induction-step} holds with
    {\allowdisplaybreaks
    \begin{align*}
        &\|\hat{x}_{\lambda|k+1} -\hat{x}^\ast_{\lambda+1|k}\| \nonumber\\
        &\ \, \stackrel{\eqref{MPC_scheme-d}}{\leq}\! \|f_{\hat{\theta}_k}(\hat{x}_{\lambda-1|k+1},\hat{u}_{\lambda-1|k+1}) - f_{\hat{\theta}_k}(\hat{x}^\ast_{\lambda|k},\hat{u}_{\lambda|k})\| \\
        &\qquad\ + \|G(\hat{x}_{\lambda-1|k+1},\hat{u}_{\lambda-1|k+1})\Delta \hat{\theta}_k\|  \\
        &\stackrel{\eqref{assump_regularity-eq-a},\eqref{assump_regularity-eq-c}}{\leq} L_\mathrm{dyn}(1+L_\kappa) \|\hat{x}_{\lambda-1|k+1}-\hat{x}^\ast_{\lambda|k}\|\\
        &\qquad\ +\|G(\hat{x}_{\lambda-1|k+1},\hat{u}_{\lambda-1|k+1}) \Delta \hat{\theta}_k\|\\
        & \stackrel{\eqref{probsetup_prop_LMS-theta-diff-eq},\eqref{prop_prelim_induction-step}}{\leq} L_\mathrm{dyn}(1+L_\kappa) \left(\sum_{i=0}^{\lambda-1}\left[L_\mathrm{dyn}(1+L_\kappa)\right]^i \|\tilde{x}_{1|k}+w_k\|\right) \\
        &\qquad +\|G(\hat{x}_{\lambda-1|k+1},\hat{u}_{\lambda-1|k+1})\| \cdot \sqrt{\mu} \|\tilde{x}_{1|k} + w_k\|\\
        &\ \, \stackrel{\eqref{probsetup_LMS_gain-definition}}{\leq} \sum_{i=1}^{\lambda}\left[L_\mathrm{dyn}(1+L_\kappa)\right]^i \|\tilde{x}_{1|k}+w_k\|+ \|\tilde{x}_{1|k} + w_k\|\\
        &\,\ =\ \sum_{i=0}^{\lambda} \left[L_\mathrm{dyn}(1+L_\kappa)\right]^i\cdot \|\tilde{x}_{1|k}+w_k\|.
    \end{align*}}
    \emph{Uniform constant: } 
    Lastly, $C_{dyn,\kappa}=C(N+1)\geq C(j),\ j\in\mathbb{I}_{[1,N+1]}$ yields~\eqref{prop_prelim-result-ineq}.
\end{proof}

\begin{lemma}[Limits of series~{\cite[Lemma~1]{solopertoGuaranteedClosedLoopLearning2023}}]\label{lemma_infinite-series}
    For any sequence $(a_k)$ that satisfies $0\leq a_k\leq a_{\max} < \infty\ \forall k\in\mathbb{N}$ and any function $\alpha \in \mathcal{K}_\infty$, the following implication holds: 
    \begin{align} \label{lemma_infinite-series_equation_appendix}
        \lim_{T\rightarrow\infty}\sum_{k=0}^{T-1} a_k \leq S < \infty \Rightarrow \lim_{T\rightarrow\infty}\sum_{k=0}^{T-1} \frac{\alpha(a_k)}{T} = 0.
    \end{align}
\end{lemma}

%%%%%%%%%%%%%%%%%%%%%%%%%%%%%%%%%%%%%%%%%%%%%%%%%%%%%%%%%%%%%%%%%
\subsection{Offline design} \label{ssec:app_SDP}
In the following, we provide details on the offline computations in Algorithm~\ref{MPC_algorithm}.
First, we discuss an optimization problem to compute a tube formulation with auxiliary feedback $\kappa$ such that Assumption~\ref{assump_Phi-function} is satisfied. Then, we show how $P_\mathrm{f}$ from Proposition~\ref{prop_constr-term-ingred} can be computed with a semidefinite program, for which we provide semi-automatic code online.

For simplicity, we consider input-independent matrices $G(x), E(x)$ in the following. This is in accordance with the numerical example, see also~\cite{sasfiRobustAdaptiveMPC2023} for the more general case.

\paragraph*{Homothetic tube and feedback $\kappa$}
We use the homothetic tubes from~\cite{kohlerComputationallyEfficientRobust2021} with ellipsoidal tubes and use ideas from~\cite{sasfiRobustAdaptiveMPC2023} for the propagation of the tube size in~\eqref{numex_tube-propagation}. The tube shape and the auxiliary feedback $\kappa(x,z,v) = K(x-z)+v$ are obtained offline by computing the matrices $P$ and $K$. Suppose that the constraint set consists of $n_c$ constraints and given as $\mathbb{Z}=\{(x,u)\in\mathbb{R}^{n_x}\times\mathbb{R}^{n_u}:\ g_i(x,u)\leq 0, i\in\mathbb{I}_{[1,n_c]}\}$.  In the following, we show how to design $P,K$ using linear matrix inequalities.  
\begin{subequations} \label{app_SDP_tubes-eq}
    \begin{alignat}{3}
    \bquad \min_{\substack{c, X, Y,\\ L_\mathrm{w}, \bar{w}, \rho}}\ & \sum_{i=0}^{n_c} \left(\frac{[c]_i\bar{w}}{(1-\rho-L_\mathrm{w})}\cdot \frac{1}{-[g(z^\mathrm{s},v^\mathrm{s})]_i}\right)^2 \label{eq:SDP_tube-a}\\
    &\text{s.t.} \nonumber\\
    & \begin{bmatrix}
                    \rho^2X&     (A(z,v,\theta)X+B(z,v,\theta)Y)^\top\\
                    \ast&        X
                    \end{bmatrix} 
                \succeq 0 
                \label{eq:SDP_tube-b}\\
              &  \begin{bmatrix}
                    X&    (G(z)[\theta-\bar{\theta}]+E(z)d)\\
                    \ast & \bar{w}^2
                    \end{bmatrix} 
                \succeq 0 
                \label{eq:SDP_tube-c}\\
              &  \begin{bmatrix}
                   L_\mathrm{w}^2 &    G_x(z,\theta)+E_x(z,d)\\
                    \ast & X
                   \end{bmatrix} 
                \succeq 0 
                \label{eq:SDP_tube-d}\\
              &  \begin{bmatrix}
                    [c]_i^2           &  \left(\left.\frac{\partial g_j}{\partial x}\right|_{(z,v)} X+ \left.\frac{\partial g_j}{\partial u}\right|_{(z,v)}Y\right)\\
                    \ast &  X
                    \end{bmatrix} \succeq 0 
                \label{eq:SDP_tube-e}\\
              &  \frac{[c]_i\bar{w}}{1-\rho-L_\mathrm{w}}\cdot \frac{1}{-[g(z^\mathrm{s}, v^\mathrm{s})]_j} \leq 1,
                \label{eq:SDP_tube-f}\\[5pt]
              &   (z,v)\in\mathbb{Z},\ \theta\in\Theta,\ d\in\mathbb{D},\ i\in \mathbb{I}_{1,n_c}, \nonumber
\end{alignat}
\end{subequations}
where $P:=X^{-1}$, $K:=Y\cdot P$, $\rho\in[0,1)$ as a contraction rate, $z^\mathrm{s}=h(v^\mathrm{s},\bar{\theta})$ is an equilibrium for the nominal parameters $\bar{\theta}$, and
\begin{align*}
G_x(z,\theta)&:= \sum_{j=1}^{n_\theta} \left.\frac{\partial [G]_{:,j}}{\partial x}\right|_{z}\cdot [\theta-\bar{\theta}]_j,\\
E_x(z,d) &:= \sum_{j=1}^{n_d} \left.\frac{\partial [E]_{:,j}}{\partial x}\right|_{z} \cdot [d]_j,
\end{align*}
The Jacobians of the dynamics are given by
\begin{align*}
    A(x,u, \theta) &:= \left.\frac{\partial f}{\partial x}\right|_{(x,u, {\theta})}, \quad 
    B(x,u, \theta) := \left.\frac{\partial f}{\partial u}\right|_{(x,u, {\theta})}.
 \end{align*}

 The linear matrix inequality (LMI)~\eqref{eq:SDP_tube-b} is an equivalent reformation of~\eqref{numex_Lyapunov_inequality}, the LMI~\eqref{eq:SDP_tube-c} ensures that $\bar{w}\geq {\max_{\theta^i+\bar{\theta}\in\Theta, d^i\in\mathbb{D}}} \|G(z)\cdot[\theta^i-\bar{\theta}] + d^i\|_P$ $\forall (z,v)\in\mathcal{Z}$ (see~\eqref{numex_tube-propagation}), and similarly~\eqref{eq:SDP_tube-d} ensures~\eqref{numex_Lw-definition}. 
 The LMI~\eqref{eq:SDP_tube-e} defines the tightening $[c]_j$ of each constraint $g_j$. 
 LMI~\eqref{eq:SDP_tube-f} checks that the tube around around $(z^\mathrm{s},v^\mathrm{s})$ lies inside the constraint set for any size $s_{j|k}\leq \frac{\bar{w}}{1-\rho-L_\mathrm{w}}$, which is needed for initial feasibility of the terminal conditions~\eqref{robust-MPC_scheme-f}--\eqref{robust-MPC_scheme-f2} in the MPC. 
The constraint~\eqref{eq:SDP_tube-f} only holds with admissible values of $\rho,L_\mathrm{w},\bar{w}$. To find suitable values, we can check the feasibility of the Problem~\eqref{app_SDP_tubes-eq} for a range of values of $\bar{w},\ L_\mathrm{w}$, and $\rho$ with standard solvers for semidefinite programs (SDP), like MOSEK~\cite{mosek}.
 The objective~\eqref{eq:SDP_tube-a} minimizes the constraint tightening at the steady-state $(z^{\mathrm{s}},v^{\mathrm{s}})$ directly. 
\begin{remark}[Polytopic constraints]
    For polytopic constraints, i.e., $\mathbb{Z}=\{H_x x+ H_u u \leq h_{xu}\}$, the LMI~\eqref{eq:SDP_tube-e} simplifies to
    \begin{align*}
        \begin{bmatrix}
        [c]_j^2           &  [H_x]_j X + [H_u]_j Y\\
        \ast &  X
        \end{bmatrix} \succeq 0.
    \end{align*}
\end{remark}
In summary, this SDP and the propagation formula~\eqref{numex_tube-propagation} ensure that the tube satisfies Assumption~\ref{assump_Phi-function}.
Furthermore, the LMIs~\eqref{eq:SDP_tube-e} and~\eqref{eq:SDP_tube-f} guarantee that the largest possible tube (with size $\frac{\bar{w}}{1-\rho-L_\mathrm{w}}$) centered around the steady-state $(z^{\mathrm{s}},v^{\mathrm{s}})$ also satisfies constraints~$\mathbb{Z}$. 

\paragraph*{Terminal cost}
To obtain the terminal cost $\ell_\mathrm{f}$ from Proposition~\ref{prop_constr-term-ingred}, and thus satisfying Assumption~\eqref{assump_terminal-ingredients}, we first solve Problem~\eqref{app_SDP_tubes-eq} to determine the feedback $K$. Second, we compute the constant $\alpha$ and the matrix $Q$ according to~\eqref{alpha_def-eq} and~\eqref{Q_def-eq}, respectively. Then, with the fixed matrix $K$, we obtain $P_\mathrm{f}$ from the following SDP 
\begin{alignat}{3}
\label{app_SDP_termcost-eq}
\min_{P_\mathrm{f}}\ & &&\!\!\mathrm{tr}(P_\mathrm{f}) \nonumber\\[-5pt]
&\text{s.t. } && A_K(x,v^\mathrm{s},\theta)^\top P_\mathrm{f} A_K(x,v^\mathrm{s},\theta) - P_\mathrm{f} \nonumber\\
            & && \qquad \qquad \qquad \preceq -Q, \nonumber\\
            & &&  \forall x\in\mathbb{Z}_x,\ v^\mathrm{s}\in\mathbb{R}^{n_u},\  \theta\in\Theta,
\end{alignat}
where $\mathrm{tr}(P_\mathrm{f})$ denotes the trace of $P_\mathrm{f}$, $A_K$ is defined in~\eqref{offline_linearization_AK-def-eq}, $\mathbb{O}_{n_x}\in\mathbb{R}^{n_x \times n_x}$ is the zero matrix.

\bibliography{Bibliography_abbrev} 
\bibliographystyle{IEEEtran} 
 \begin{IEEEbiography}[{\includegraphics[trim=0cm 5cm 0cm 0cm, width=1in,height=1.25in,clip,keepaspectratio]{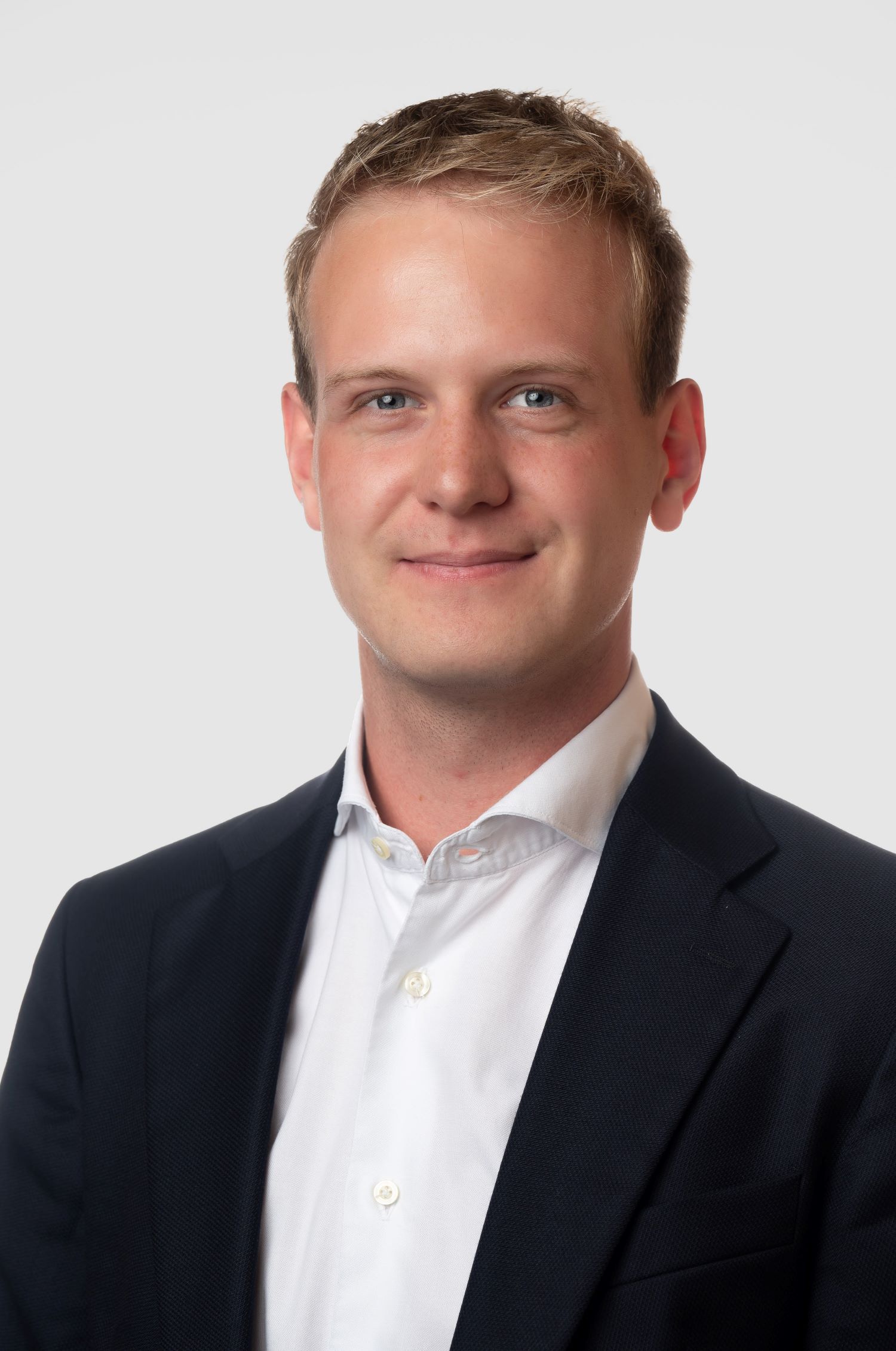}}]{Maximilian Degner}
received his B.Sc. and his M.Sc. with distinction in mechanical engineering from ETH Zürich, Switzerland, in 2022 and 2024, respectively. Since 2025, he is a PhD student at the Institute for Systems Theory and Automatic Control, University of Stuttgart.
During his studies, he focused on optimization-based control methods.
His research interests lie at the intersection of online learning and control.
 \end{IEEEbiography}
\vskip -1\baselineskip plus -1fil

 \begin{IEEEbiography}[{\includegraphics[ width=1in,clip,keepaspectratio]{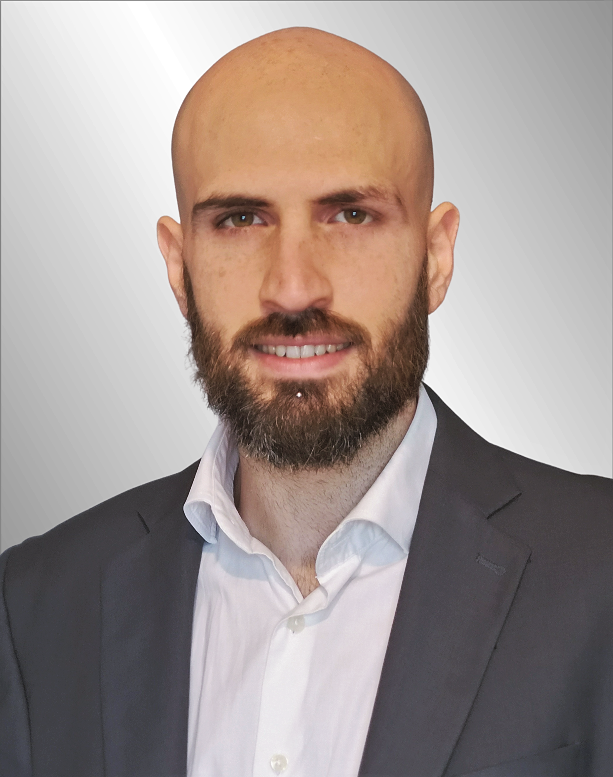}}]{Raffaele Soloperto}
received his Bachelor’s and Master’s degrees in Automation Engineering from the University of Bologna, Italy, in 2014 and 2016, respectively. He earned his Ph.D. in 2022 from the University of Stuttgart, Germany, in collaboration with the International Max Planck Research School (IMPRS).
From 2022 to 2025, he was a Postdoctoral Researcher at the Automatic Control Laboratory, ETH Zürich, Switzerland. Since 2025, he has been serving as Innovation Manager at Embotech AG. His research interests include model predictive control and game theory.
\end{IEEEbiography}
\vskip -1\baselineskip plus -1fil

\begin{IEEEbiography}[{\includegraphics[width=1in,height=1.25in,clip,keepaspectratio]{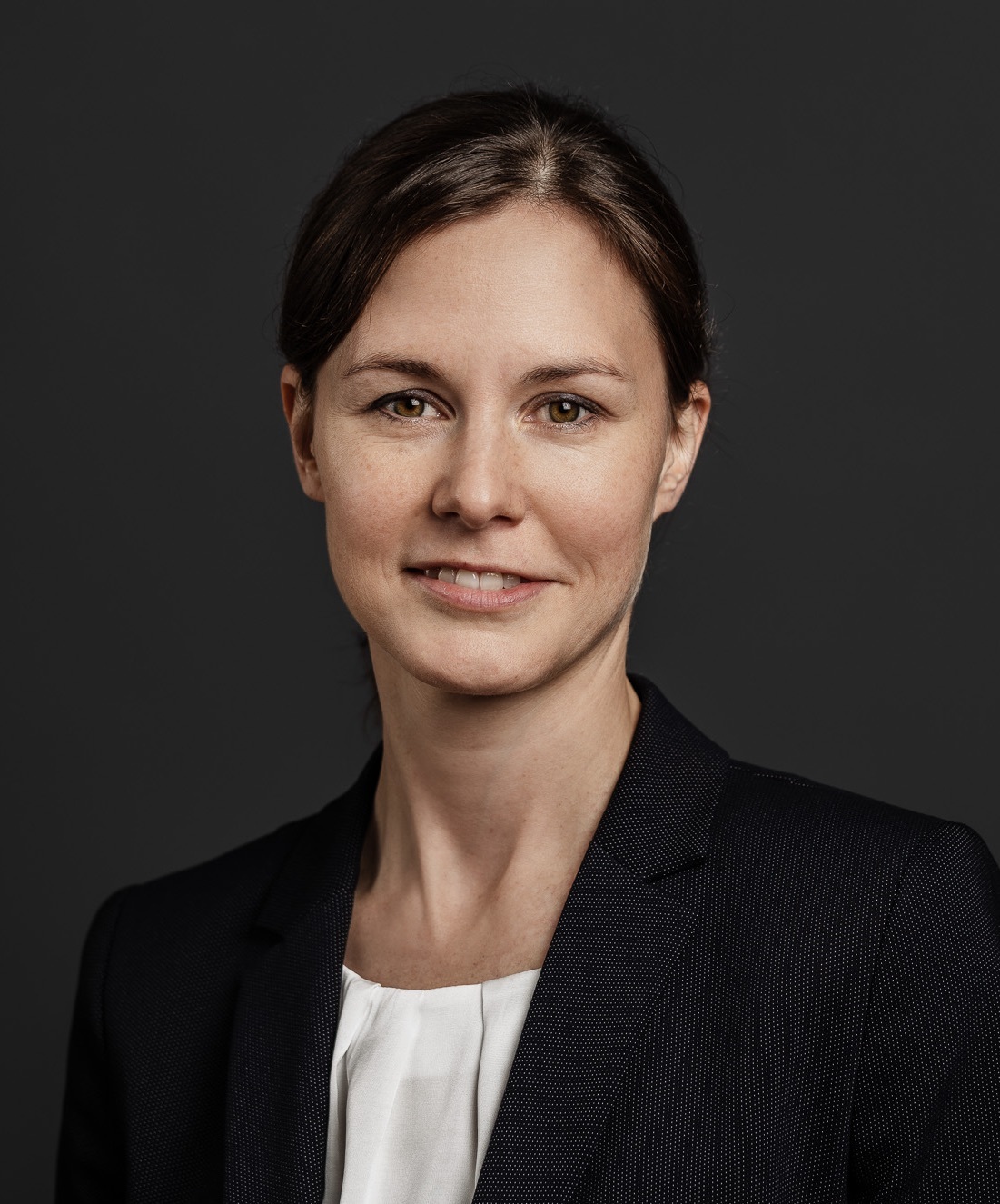}}]{Melanie N. Zeilinger}
 is an Associate Professor at ETH Zürich, Switzerland. 
She received the Diploma degree in engineering cybernetics from the University of Stuttgart, Germany, in 2006, and the Ph.D. degree with honors in electrical engineering from ETH Zürich, Switzerland, in 2011. 
From 2011 to 2012 she was a Postdoctoral Fellow with the École Polytechnique Fédérale de Lausanne (EPFL), Switzerland.
She was a Marie Curie Fellow and Postdoctoral Researcher with the Max Planck Institute for Intelligent
Systems, Tübingen, Germany until 2015 and with the Department of Electrical Engineering and Computer Sciences at the University
of California at Berkeley, CA, USA, from 2012 to 2014. 
From 2018 to 2019 she was a professor at the University of Freiburg, Germany. 
Her current research interests include safe learning-based control, as well as distributed control and optimization, with applications to robotics and human-in-the loop control.
\end{IEEEbiography}
\vskip -1\baselineskip plus -1fil

\begin{IEEEbiography}[{\includegraphics[trim=0cm 5cm 0cm 0cm, width=1in,height=1.25in,clip,keepaspectratio]{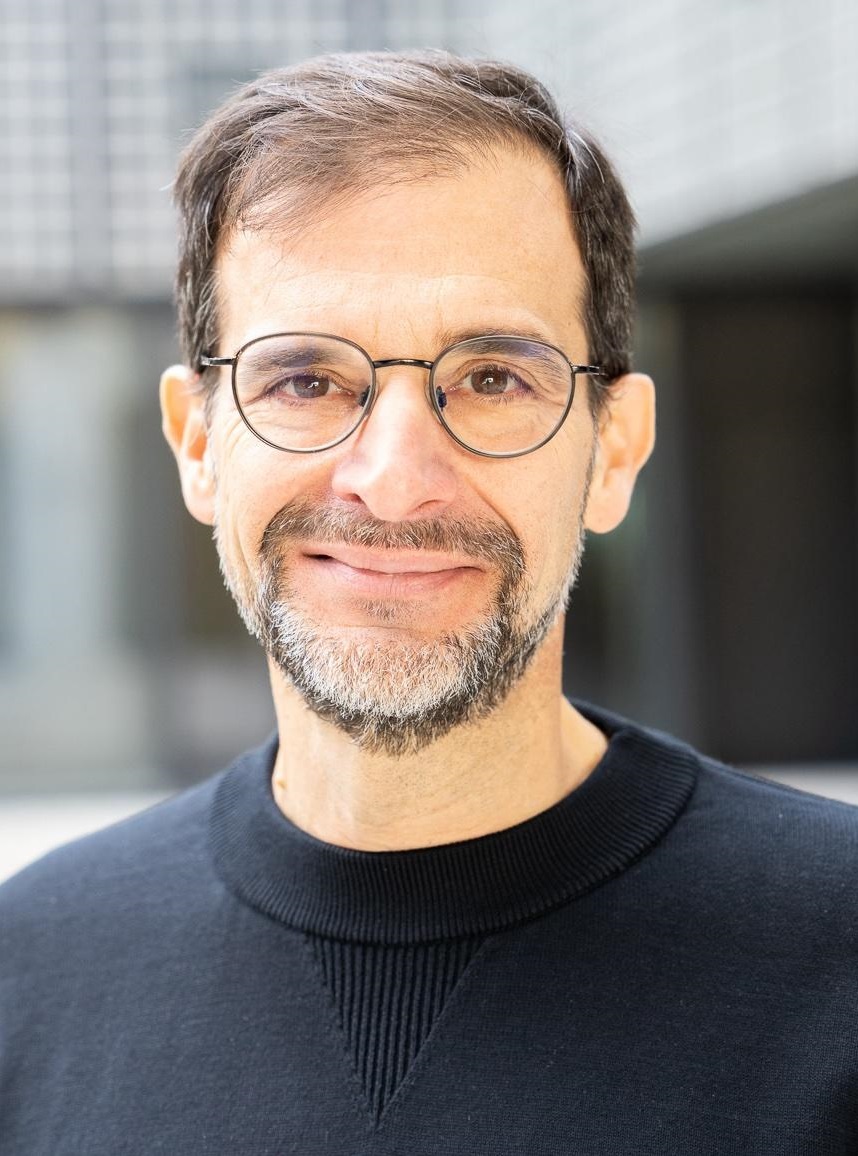}}]{John Lygeros}
completed a B.Eng. degree in electrical
engineering in 1990 and an M.Sc. degree in Systems
Control in 1991, both at Imperial College of Science
Technology and Medicine, London, U.K. In 1996 he
obtained a Ph.D. degree from the Electrical Engineering and Computer Sciences Department, University of
California, Berkeley. During the period 1996–2000 he
held a series of research appointments at the National
Automated Highway Systems Consortium, Berkeley, the
Laboratory for Computer Science, M.I.T., and the Electrical Engineering and Computer Sciences Department
at U.C. Berkeley. Between 2000 and 2003 he was a University Lecturer at
the Department of Engineering, University of Cambridge, U.K., and a Fellow
of Churchill College. Between 2003 and 2006 he was an Assistant Professor at
the Department of Electrical and Computer Engineering, University of Patras,
Greece. In July 2006 he joined the Automatic Control Laboratory at ETH Zürich,
first as an Associate Professor, and since January 2010 as a Full Professor;
he is currently serving as the Head of the laboratory and the head of the
Department of Information Technology and Electrical Engineering. His research
interests include modeling, analysis, and control of hierarchical, hybrid, and
stochastic systems, with applications to biochemical networks, transportation
systems, energy systems, and industrial processes. John Lygeros is a Fellow
of the IEEE, and a member of the IET and the Technical Chamber of Greece;
between 2013 and 2023 he served as the Vice President for Finances and a
Council Member of the International Federation of Automatic Control (IFAC), as
well as the Board of the IFAC Foundation.
\end{IEEEbiography}
\vskip -1\baselineskip plus -1fil
 
\begin{IEEEbiography}[{\includegraphics[width=1in,height=1.25in,clip,keepaspectratio]{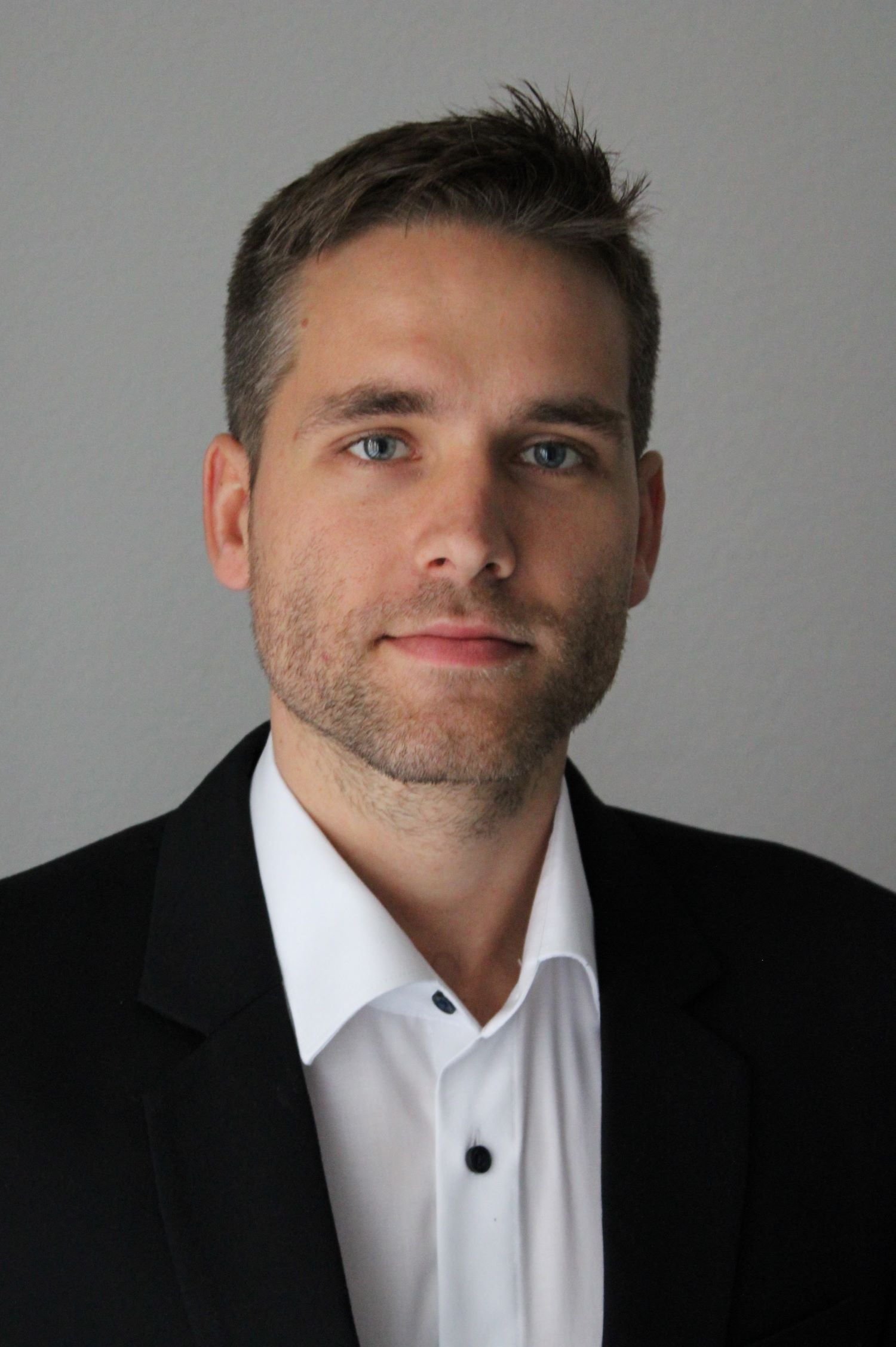}}]{Johannes K\"ohler} \looseness -1  is an Assistant Professor at Imperial College London. 
He received the Ph.D. degree from the University of Stuttgart, Germany, in 2021. 
From 2021 to 2025, he was a postdoctoral researcher at ETH Zurich, Switzerland.
He has received several awards including the 2021 European Systems \& Control PhD Thesis Award, the IEEE CSS George S. Axelby Outstanding Paper Award 2022, and the Journal of Process Control Paper Award 2023. 
His research interests include data-driven models and predictive control with applications to robotics, autonomous systems, and biomedical problems. 
\end{IEEEbiography}
\end{document}